%% file: main.tex
\documentclass{article}

\usepackage{microtype}
\usepackage{graphicx}
\usepackage{subfigure}
\usepackage{booktabs} 

\usepackage{mathtools}
\usepackage{paralist}
\usepackage{amsmath,amssymb,amsthm}

\newtheorem{theorem}{Theorem}[section]
\newtheorem{lemma}[theorem]{Lemma}
\newtheorem{definition}[theorem]{Definition}
\newtheorem{proposition}[theorem]{Proposition}
\newtheorem{claim}[theorem]{Claim}
\newtheorem{corollary}[theorem]{Corollary}

\theoremstyle{definition}
\newtheorem{example}{Example}[section]
\newtheorem{problem}{Problem}[section]
\newtheorem{remark}{Remark}[section]
\usepackage[stable]{footmisc}
\newcommand{\R}{\mathbb{R}}
\newcommand{\x}{\mathbf{x}}
\newcommand{\mb}[1]{\mathbf{#1}}
\newcommand{\mc}[1]{\mathcal{#1}}

\newcommand{\argmax}{\mathop{\mathrm{argmax}}}

\newcommand{\ip}[1]{\langle{ #1 \rangle} }
\newcommand{\epsscore}{\epsilon'} 
\newcommand{\epsprivacy}{\epsilon}

\allowdisplaybreaks

\usepackage{hyperref}

\usepackage[accepted]{icml2020}

\icmltitlerunning{Fast and Private Submodular and $k$-Submodular Functions Maximization}

\begin{document}

\onecolumn
\icmltitle{Fast and Private Submodular and $k$-Submodular Functions Maximization with Matroid Constraints}



\icmlsetsymbol{equal}{*}

\begin{icmlauthorlist}
\icmlauthor{Akbar Rafiey}{to}
\icmlauthor{Yuichi Yoshida}{goo}
\end{icmlauthorlist}

\icmlaffiliation{to}{Department of Computing Science, Simon Fraser University, Burnaby, Canada}
\icmlaffiliation{goo}{National Institute of Informatics, Tokyo, Japan}

\icmlcorrespondingauthor{Akbar Rafiey}{arafiey@sfu.ca}
\icmlcorrespondingauthor{Yuichi Yoshida}{yyoshida@nii.ac.jp}

\icmlkeywords{Differential Privacy, Submodular Maximization}

\vskip 0.3in



\printAffiliationsAndNotice{}  


\begin{abstract}
The problem of maximizing nonnegative monotone submodular functions under a
certain constraint has been intensively studied in the last decade, and a wide range of efficient approximation algorithms have been developed for this problem. Many machine learning problems, including data summarization and influence maximization, can be naturally modeled
as the problem of maximizing monotone submodular functions. However, when such applications involve
sensitive data about individuals, their privacy
concerns should be addressed. In this paper, we study the problem of maximizing monotone submodular functions subject to matroid constraints in the
framework of differential privacy. We provide $(1-\frac{1}{\mathrm{e}})$-approximation algorithm which improves upon the previous results in terms of approximation guarantee. This is done with an almost cubic number of function evaluations in our algorithm.

Moreover, we study $k$-submodularity, a natural generalization of submodularity. We give the first $\frac{1}{2}$-approximation algorithm that preserves
differential privacy for maximizing monotone $k$-submodular
functions subject to matroid constraints. The approximation ratio is asymptotically tight and is obtained with an almost linear number of function evaluations.
\end{abstract}

\input{intro}

\input{pre}

\input{continuous-greedy}

\section{Improving the Query Complexity}\label{section:IDPCG}
In this section, we improve the number of evaluations of $F$ from $O(r(\mc{M})|E|^{1+({\frac{r(\mc{M})}{\epsprivacy}})^2})$ to $O(r(\mc{M})|E|^2\ln{\frac{|E|}{\epsprivacy}})$. In Algorithm~\ref{alg:DPGA}, in order to choose a point with probability proportional to $\exp(\ip{ \mb{y},\nabla f(\mb{x})})$, it requires to compute $Z=\sum\limits_{\mb{z}\in C_{\rho}}\exp(\ip{ \mb{z},\nabla f(\mb{x})})$. This summation needs evaluating $(\ip{ \mb{z},\nabla f(\mb{x})})$ for all $\mb{z}$ in $C_{\rho}$. One way of improving the query complexity of this step is as follows. Partition $C_{\rho}$ into a number of layers such that points in each layer are almost the same in terms of the inner product $\ip{ \cdot ,\nabla f(\mb{x})}$.
Now, instead of choosing a point in $C_{\rho}$, we carefully select a layer with some probability (i.e., proportional to its size and \emph{quality} of points in it) and then choose a point from that layer uniformly at random. Of course, to estimate the size of each layer, we need to sample a sufficiently large number of points from $C_{\rho}$.


\begin{definition}[layer]\label{layer}
  For a point $\mb{x}\in C_{\rho}$ and $\mu > 0$, let the $i$-th layer to be
  $\mathcal{L}^{\mb{x}}_{\mu,i}=\{\mb{z}\in C_{\rho} \mid {(1+\mu)}^{i-1}\leq\exp\bigl(\ip{\mb{z},\nabla f(\mb{x})}\bigr)< {(1+\mu)}^i\}$,
  for $1\leq i\leq k$, where
  \[
    k= \left\lceil \log\limits_{1+\mu}\left(\frac{\max\limits_{\mb{y}\in C_{\rho}}\exp\bigl(\ip{ \mb{y},\nabla f(\mb{x})}\bigr)}{\min\limits_{\mb{y}\in C_{\rho}}\exp\bigl(\ip{ \mb{y},\nabla f(\mb{x})}\bigr)}\right)\right\rceil.
  \]
\end{definition}
For a layer $\mathcal{L}^{\mb{x}}_{\mu,i}$ let $|\mathcal{L}^{\mb{x}}_{\mu,i}|$ denote the number of points in it, and define $\tilde{Z} \in \R$ and $\tilde{Z}_i \in \R$ for each $i \in [k]$ as follows:
\[
  \tilde{Z}=\sum\limits_{i\in[k]}|\mathcal{L}^{\mb{x}}_{\mu,i}|(1+\mu)^{i-1}
  \qquad \text{and} \qquad
  \tilde{Z}_i=|\mathcal{L}^{\mb{x}}_{\mu,i}|(1+\mu)^{i-1}.
\]

Then, a layer $\mathcal{L}^{\mb{x}}_{\mu,i}$ is chosen with probability $\frac{\tilde{Z}_i}{\tilde{Z}}$.
Note that we do not want to spend time computing the exact value of $|\mathcal{L}^{\mb{x}}_{\mu,i}|$ for every layer, instead, we are interested in efficiently estimating these values. 
By Hoeffding's inequality \cite{hoeffding1963probability}, to estimate $|\mathcal{L}^{\mb{x}}_{\mu,i}|/|C_{\rho}|$ with additive error of $\lambda$ with probability at least $1-\theta$, it suffices to sample $\Theta( \ln(1/\theta)/\lambda^2)$ points from $C_{\rho}$. Hence, by a union bound, if we want to estimate $|\mathcal{L}^{\mb{x}}_{\mu,i}|/|C_{\rho}|$ with additive error of $\lambda$ for all $ i=1,\dots,k$ with probability at least $1-\theta$, it suffices to sample $\Theta(\ln(k/\theta)/\lambda^2)$ points from $C_{\rho}$.

\begin{algorithm}[tb]
  \caption{Improved Differentially Private Continuous Greedy Algorithm}
  \label{alg:IDPGA}
\begin{algorithmic}[1]
\STATE {\bfseries Input:} {Submodular function $F_D\colon 2^E \to [0,1]$, dataset $D$, a matroid $\mathcal{M}=(E,\mathcal{I})$, and $\epsprivacy,\mu,\rho, \lambda, \theta > 0$.}
\STATE Let $C_{\rho}$ be a $\rho$-covering of $\mathcal{P}(\mathcal{M})$, and $f_D$ be the multilinear extension of $F_D$.
\STATE $\mb{x}(0)\gets \mathbf{0}$, $\epsscore\gets \frac{\epsprivacy}{2\Delta}$.
\FOR{$t=1$ to $T=r(\mc{M})$}
    \STATE $C'_{\rho}\gets$ Sample $\Theta(\ln(k/\theta)/\lambda^2)$ points from $C_{\rho}$ uniformly at random. \label{sampling}
	\STATE Define $\mathcal{L}^{\mb{x}_{t-1}}_{\mu,i}$ as in Definition~\ref{layer}, and estimate each $|\mathcal{L}^{\mb{x}_{t-1}}_{\mu,i}|$ using $C'_{\rho}$.
	\STATE Let $\tilde{L}^{\mb{x}_{t-1}}_{\mu,i}$ denote the estimated value.
    \STATE Set $\tilde{Z}_i\gets\tilde{L}^{\mb{x}_{t-1}}_{\mu,i}(1+\mu)^{\epsscore (i-1)}$ and $
    \tilde{Z}\gets \sum\limits_{i\in[k]}\tilde{L}^{\mb{x}_{t-1}}_{\mu,i}(1+\mu)^{\epsscore (i-1)}$
	\STATE Let $\mathcal{L}$ be the chosen layer $\mathcal{L}^{\mb{x}_{t-1}}_{\mu,i}$ with probability 			proportional to $\frac{\tilde{Z}_i}{\tilde{Z}}$.
	\STATE Let $\mb{y}_{t-1}$ be a point sampled uniformly at random from $\mathcal{L}$.
	\STATE $\mb{x}_t\gets \mb{x}_{t-1}+\alpha\mb{y}_{t-1}$.
\ENDFOR
    \STATE {\bfseries Outout:} $\mb{x}_T$
  \end{algorithmic}
\end{algorithm}

\begin{corollary}\label{estimation-error}
  Let $C_{\rho}$ be a $\rho$-covering of $\mc{P}(\mc{M})$ and $\mb{x}_t$ be a point in $\mc{P}(\mc{M})$. Algorithm~\ref{alg:IDPGA} estimates $|\mathcal{L}^{\mb{x}_t}_{\mu,i}|/|C_{\rho}|$ with an additive error $\lambda_{\ref{estimation-error}}$ with probability at least $1-\theta_{\ref{estimation-error}}$.
\end{corollary}

\begin{lemma}[Analogous to Theorem~\ref{thm:EM-bound}]\label{lemma:EM-Analogous}
At each time step $t$, Algorithm~\ref{alg:IDPGA} returns $\mb{y}_{t-1}$ such that for every $\beta\in (0,1)$ and $\xi=\ln\left(\frac{|C_{\rho}|(1+k\lambda |C_{\rho}|)(1+\mu)^{\epsscore}}{\beta}\right)$ we have
\begin{align*}
    \Pr
    \left[\ip{\mb{y}_{t-1},\nabla f(\mb{x}_{t-1})} \geq \max_{\mb{z} \in C_{\rho}} \ip{\mb{z},\nabla f(\mb{x}_{t-1})}-\frac{2\Delta}{\epsprivacy}\xi\right] 
    \geq 1-\beta.
 \end{align*}
\end{lemma}

\begin{theorem}\label{IDPGA-optimality}
  Suppose $F_D$ is $\Delta$-sensitive and $C_\rho$ is a $\rho$-covering of $\mc{P}(\mc{M})$. Then Algorithm~\ref{alg:IDPGA}, with high probability (depending on $\theta_{\ref{estimation-error}}$), returns $\mb{x}_T\in\mc{P}(\mc{M})$ such that
  \begin{align*}
  f(\mb{x}_T)
  \geq
  \left(1-\frac{1}{e}\right)\mathrm{OPT} -  O\Bigg(C_{\ref{lem:existance-in-net}}\sqrt{\rho}  + \ln(1+\mu)+ 
  \left(\frac{\Delta r(\mc{M})}{\epsprivacy\rho^2}\right)\left(\ln|E|+\ln(k\lambda_{\ref{estimation-error}})\right)\Bigg) 
  \end{align*}
\end{theorem}

\begin{theorem}\label{IDPGA-privacy}
  Algorithm~\ref{alg:IDPGA} preserves
  $O\left( \epsprivacy r(\mc{M})^2  \right)$-differential privacy.
\end{theorem}

\begin{theorem}[Formal version of Theorem~\ref{thm:imp-introduction}]
Suppose $F_D$ is $\Delta$-sensitive and Algorithm~\ref{alg:IDPGA} is instantiated with $\rho=\frac{\epsprivacy}{|E|^{1/2}},\mu=\mathrm{e}^{\epsprivacy}, \lambda_{\ref{estimation-error}}=1/\sqrt{|E|}, \theta_{\ref{estimation-error}}=1/|E|^2$. Then Algorithm~\ref{alg:IDPGA} is $(\epsprivacy r(\mc{M})^2)$-differentially private and, with high probability, returns $S\in \mathcal{I}$ with quality at least
\[
  F_D(S) \geq \left(1-\frac{1}{e}\right)\mathrm{OPT} - O\left(\sqrt{\epsprivacy}+ \frac{\Delta r(\mc{M})|E|\ln( \frac{|E|}{\epsprivacy})}{\epsprivacy^3}\right).
  \]
  Moreover, it evaluates $F_D$ at most $O(r(\mc{M})|E|^2\ln(\frac{|E|}{\epsprivacy}))$ times.
\end{theorem}

\input{k-submodular}
\input{conclusion}
\section*{Acknowledgments}
A.R.\ is thankful to Igor Shinkar and Nazanin Mehrasa for useful discussions. We also thank anonymous referees for useful suggestions. A.R.\ is supported by NSERC. 
Y.Y.\ is supported by JSPS KAKENHI Grant Number 18H05291.

\bibliographystyle{icml2020}
\bibliography{main}
\onecolumn
\appendix
\input{appendix.tex}

\end{document}

%% file: intro.tex

\section{Introduction}

A set function $F\colon 2^E \to \mathbb{R}$ is \emph{submodular} if for any $S \subseteq T \subseteq E$ and $e \in E\setminus T$ it holds that
$
F(S \cup \{e \}) - F(S) \geq F(T \cup \{e\}) - F(T).
$
The theory of \emph{submodular maximization} provides a general and unified framework for various combinatorial optimization problems including the Maximum Coverage, Maximum Cut, and Facility Location
problems. Furthermore, it also appears in a wide variety of applications such as viral marketing~\cite{kempe2003maximizing}, information gathering~\cite{krause2007near}, feature selection for classification~\cite{krause2005near}, influence maximization in social networks~\cite{kempe2003maximizing}, document summarization~\cite{lin2011class}, and speeding up satisfiability solvers~\cite{streeter2009online}. For a survey,  see~\cite{krause2014submodular}. As a consequence of these applications and importance, a wide range of efficient approximation algorithms have been developed for maximizing submodular functions subject to different constraints~\cite{calinescu2011maximizing,nemhauser1978best,nemhauser1978analysis,vondrak2008optimal}. 

The need for efficient optimization methods that guarantee the privacy of individuals is wide-spread across many applications concerning sensitive data about individuals, e.g., medical data, web search query data, salary data, social networks. Let us motivate privacy concerns by an example.

\begin{example}[Feature Selection~\cite{krause2005near,MitrovicB0K17}] \label{example-1}
A sensitive dataset $D=\{(\mb{x}_i,C_i)\}_{i=1}^{n}$ consists of a feature vector $\mb{x}_i=(\mb{x}_i(1),\dots,\mb{x}_i(m))$ associated to each individual $i$ together with a binary class label $C_i$. The objective is to select a small (e.g., size at most $k$) subset $S\subseteq [m]$ of features that can provide a good classifier for $C$. One particular example for this setting is determining collection of features such as height, weight, and age that are most relevant in predicting if an individual is likely to have a particular disease such as diabetes and HIV. One approach to address the feature selection problem, due to~\citet{krause2005near}, is based on maximizing a submodular function which captures the mutual information between a subset of features and the class label of interest. Here, it is important that the selection of relevant features does not compromise the privacy of any individual who has contributed to the training dataset.
 \end{example}

\emph{Differential privacy} is a rigorous notion of privacy that allows statistical analysis of sensitive data while providing strong
privacy guarantees. Basically, differential privacy requires that computations be insensitive
to changes in any particular individual's record. 
A dataset is a collection of records from some domain, and two datasets are \textit{neighboring} if they differ in a single record. Simply put, the requirement for differential privacy is that the computation behaves nearly identically on two neighboring datasets; 
Formally, for $\epsilon,\delta \in \R_+$, we say that a randomized computation $M$ is \emph{$(\epsilon,\delta)$-differentially private} if for any neighboring datasets $D \sim D'$, and for any set of outcomes $S \subseteq \mathrm{range}(M)$,
  \[
    \Pr[M(D)\in S] \leq \exp(\epsilon) \Pr[M(D')\in S]+\delta.
  \]
  When $\delta=0$, we say $M$ is \emph{$\epsilon$-differentially private}. Differentially private algorithms must be calibrated to the \emph{sensitivity} of the function of interest with respect to small changes in the input dataset.

In this paper we consider designing a differentially private algorithm for maximizing nonnegative and \emph{monotone} submodular functions in \emph{low-sensitivity} regime. Whilst, a \emph{cardinality} constraint (as in Example~\ref{example-1}) is a natural one to place on a submodular maximization problem, many other problems, e.g., personalized data summarization~\cite{mirzasoleiman2016fasta}, require the use of more general types of constraints, i.e., \emph{matroid} constraints. The problem of maximizing a submodular function under a matroid constraint is a
classical problem~\cite{edmonds71}, with many important special cases, e.g., uniform matroid (the
subset selection problem, see Example~\ref{example-1}), partition matroid (submodular welfare/partition problem). We consider the following.
\begin{problem}\label{problem1}
Given a sensitive dataset $D$ associated to a monotone submodular  function $F_D\colon 2^E\to \mathbb{R}_+$ and a matroid $\mc{M}=(E,\mc{I})$. Find a subset $S \in \mc{I}$ that approximately maximizes $F_D$ in a manner that guarantees differential privacy with respect to the input dataset $D$.
\end{problem}

Furthermore, we consider a natural generalization of submodular functions, namely, $k$-submodular functions. $k$-submodular function maximization allows for richer problem structure than submodular maximization. For instance, coupled feature selection~\cite{SinghB12}, sensor placement with $k$ kinds of measures~\cite{ohsaka2015monotone},  and influence maximization with $k$ topics can be expressed as $k$-submodular function maximization problems. To motivate the privacy concerns, consider the next example. More examples are given in Section~\ref{section:k-sub-examples}.

\begin{example}[Influence Maximization with $k$ Topics]
For $k$ topics, a sensitive dataset is a directed graph $G=(V,E)$ with an edge probability $p_{u,v}^{i}$ for each edge $(u,v)\in E$, representing the strength of influence from $u$ to $v$ on the $i$-th topic. The goal is to distribute these topics to $N$ vertices of the graph so that we maximize \emph{influence spread}. The problem of maximizing influence spread can be formulated as $k$-submodular function maximization problem~\cite{ohsaka2015monotone}. An example for this setting is in viral marketing where dataset consists of a directed graph where each vertex represents a user and each edge represents the friendship between a pair of users. Given $k$ kinds of products, the objective is to promote products by giving (discounted) items to a selected group of influential people in the hope that large number of product adoptions will occur. 
Here, besides maximizing the influence spread, it is important to preserve the privacy of individuals in the dataset.
\end{example}


\begin{problem}\label{problem2}
Given a sensitive dataset $D$ associated to a monotone $k$-submodular function $F_D\colon (k+1)^E\to \mathbb{R}_+$ and a matroid $\mc{M}=(E,\mc{I})$.
   Find $S = (S_1, \dots, S_k)$ with $\bigcup_{i\in[k]} S_i \in \mc{I}$ that approximately maximizes $F_D$ in a manner that guarantees differential privacy with respect to the input dataset $D$.
\end{problem}

\subsection{Our Contributions}
 
\noindent  \textbf{Submodular Maximization:} For maximizing a nonnegative monotone submodular function subject to a matroid constraint, we show that a modification of the \emph{continuous greedy} algorithm~\cite{calinescu2011maximizing} yields a good approximation guarantee as well as a good privacy guarantee. Following the same idea, we maximize the so-called \emph{multilinear extension} of the input submodular function in the corresponding \emph{matroid polytope}, denoted by $\mc{P}(\mc{M})$. However, in order to greedily choose a direction, it requires to have a \emph{discretization} of the matroid polytope. Fortunately, due to \citet{yoshida2019cheeger}, an efficient discretization can be achieved. That is, we can \emph{cover} a polytope with a small number of balls in polynomial time. Having these in hand, we prove the following. 

\begin{theorem}\label{thm:main-introduction}
Suppose $F_D$ is monotone with sensitivity $\Delta$ and $\mc{M}=(E,\mc{I})$ is a matroid. For every $\epsprivacy>0$, there is an $(\epsprivacy r(\mc{M})^2)$-differentailly private algorithm that, with high probability, returns $S\in \mc{I}$ with quality at least $(1-\frac{1}{\mathrm{e}})OPT-O\left(\sqrt{\epsprivacy}+\frac{\Delta r(\mc{M})|E|\ln{|E|}}{\epsprivacy^3}\right)$.
\end{theorem}

  For covering $C$ of $\mc{P}(\mc{M})$, the algorithm in Theorem~\ref{thm:main-introduction} makes $O(r(\mc{M})|E||C|)$ queries to the \emph{evaluation oracle}. We point out that $C$ has a size of roughly $|E|^{1/\epsprivacy^2}$. In Section~\ref{section:IDPCG}, we present an algorithm that makes significantly fewer queries to the evaluation oracle.

  \begin{theorem}\label{thm:imp-introduction}
      Suppose $F_D$ is monotone and has sensitivity $\Delta$ and $\mc{M}=(E,\mc{I})$ is a matroid. For every $\epsprivacy>0$, there is an $(\epsprivacy r(\mc{M})^2)$-differentailly private algorithm that, with high probability, returns $S\in \mc{I}$ with quality at least $(1-\frac{1}{\mathrm{e}})OPT-O\left(\sqrt{\epsprivacy}+ \frac{\Delta r(\mc{M})|E|\ln( |E|/\epsprivacy)}{\epsprivacy^3}\right)$. Moreover, this algorithm makes at most $O(r(\mc{M})|E|^2\ln{\frac{|E|}{\epsprivacy}})$ queries to the evaluation oracle.
  \end{theorem}

\noindent  \textbf{$k$-submodular Maximization:}
To the best of our knowledge, there is no algorithm for maximizing $k$-submodular functions concerning differential privacy.
We study Problem~\ref{problem2} in Section~\ref{section;k-submodular}.
First, we discuss an $(\epsprivacy r(\mc{M}))$-differentially private algorithm that uses the evaluation oracle at most $O(k r(\mc{M}) |E|)$ times and outputs a solution with quality at least $1/2$ of the optimal one.

\begin{theorem}
  Suppose $F_D:(k+1)^E\to \mathbb{R}_+$ is monotone and has sensitivity $\Delta$. For any $\epsilon > 0$, there is an $O(\epsprivacy r(\mc{M}))$-differentially private algorithm that, with high probability, returns a solution $X=(X_1,\ldots,X_k) \in (k+1)^E$ with $\bigcup_{i\in[k]} X_i \in \mc{I}$ and $F_D(X) \geq \frac{1}{2}\mathrm{OPT}-O(\frac{\Delta r(\mc{M})\ln{|E|}}{\epsprivacy})$ by evaluating $F_D$ at most $O(k r(\mc{M})|E|)$ times.
\end{theorem}
This $1/2$ approximation ratio is asymptotically tight due to the hardness result in~\cite{iwata2016improved}.
Applying a sampling technique~\cite{MirzasoleimanBK15,MitrovicB0K17,ohsaka2015monotone}, we propose an algorithm that preserves the same privacy guarantee and the same quality as before while evaluating $F_D$ almost linear number of times, namely $O\left(k|E|\ln{r(\mc{M})} \ln{\frac{r(\mc{M})}{\gamma}} \right)$. Here, $\gamma$ is the failure probability of our algorithm.


\subsection{Related Works}
\citet{gupta2010differentially} considered an important case of Problem~\ref{problem1} called the \emph{Combinatorial Public Projects} (CPP problem). 
The CPP problem was introduced by \citet{PapadimitriouSS08} and is as follows. For a data set $D=(x_1,\ldots,x_n)$, each individual $x_i$ submits a \emph{private} non-decreasing and submodular valuation function $F_{x_i}\colon 2^E\to [0, 1]$. Our goal is to select a subset $S\subseteq E$ of size $k$ to maximize function $F_D$ that takes the particular form $F_D(S) =\frac{1}{n} \sum\limits_{i=1}^{n} F_{x_i}(S)$. Note that in this setting, the sensitivity can be always bounded from above by $\frac{1}{n}$.
Gupta~\emph{et~al.} showed the following.


\begin{theorem}[\citet{gupta2010differentially}]\label{Gupta}
For any $\delta\leq 1/2$, there is an $(\epsprivacy,\delta)$-differentially private algorithm for the CPP problem under cardinality constraint that, with high probability, returns a solution $S\subseteq E$ of size $k$ with quality at least $(1-\frac{1}{\mathrm{e}})\mathrm{OPT}-O(\frac{k\ln{(\mathrm{e}/\delta)}\ln |E|}{\epsprivacy})$.
\end{theorem}
 There are many cases which do not fall into the CPP framework. For some problems, including feature selection via mutual information (Example~\ref{example-1}), the submodular function $F_D$ of interest depends on the dataset $D$ in ways much more complicated than averaging functions associated to each individual.
 Unfortunately, the privacy analysis of Theorem~\ref{Gupta} heavily relies on the assumption that the input function $F_D=\frac{1}{n} \sum_{i=1}^{n} F_{x_i}(S)$ is the average of $F_{x_i}$'s, and does not directly generalize to arbitrary submodular functions. Using a \emph{composition theorem} for differentially private mechanisms, \citet{MitrovicB0K17} proved the following 
 \begin{theorem}[\citet{MitrovicB0K17}]
 Suppose $F_D$ is monotone and has sensitivity $\Delta$. For any $\epsprivacy>0$, there is a $(k\epsprivacy)$-differentially private algorithm  that, with high probability, returns $S\subseteq E$ of size $k$ with quality at least 
 $
 \left(1-\frac{1}{\mathrm{e}}\right)\mathrm{OPT}-O\left(\frac{\Delta k\ln |E|}{\epsprivacy}\right).
 $
 \end{theorem}
 In the same work, \citet{MitrovicB0K17} considered matroid constraints and more generally \emph{$p$-extendable} constraints. 
\begin{theorem}[\citet{MitrovicB0K17}]
 Suppose $F_D$ is monotone with sensitivity $\Delta$ and let $\mc{M}=(E,\mc{I})$ be a matroid. Then for any $\epsprivacy>0$, there is an $(\epsprivacy r(\mc{M}))$-differentially private algorithm that, with high probability, returns a solution $S\in \mc{I}$ with quality at least
 $
 \frac{1}{2}\mathrm{OPT}-O\left(\frac{\Delta r(\mc{M})\ln |E|}{\epsprivacy}\right).
 $
\end{theorem}

\noindent  \textbf{$k$-submodular Maximization:} The terminology for $k$-submodular functions was first introduced in~\cite{HuberK12} while the concept has been studied previously in~\cite{CohenCJK06}.
Note for $k=1$ the notion of $k$-submodularity is the same as submodularity. For $k=2$, this notion is known as \emph{bisubmodularity}. Bisubmodularity arises in bicooperative games~\cite{bilbao2008survey} as well as variants of sensor placement problems and coupled feature selection problems~\cite{SinghB12}.
For unconstrained nonnegative $k$-submodular maximization,~\citet{ward2014maximizing} proposed a $\max\{1/3,1/(1+a)\}$-approximation algorithm where $a=\max\{1,\sqrt{(k-1)/4}\}$.
The approximation ratio was improved to $1/2$ by~\citet{iwata2016improved}.
They also provided $k/(2k-1)$-approximation for maximization of monotone $k$-submodular functions.
The problem of maximizing a monotone $k$-submodular function was considered by~\citet{ohsaka2015monotone} subject to different constraints.
They gave a $1/2$-approximation algorithm for total size constraint, i.e., $|\bigcup_{i\in[k]} X_i|\leq N$, and $1/3$-approximation algorithm for individual size constraints, i.e., $|X_i|\leq N_i$ for $i=1,\dots,k$. \citet{sakaue2017maximizing} proved that $1/2$-approximation can be achieved for matroid constraint, i.e., $\bigcup_{i\in[k]} X_i\in \mc{I}$.


%% file: pre.tex

\section{Preliminaries}\label{sec:pre}

For a set $S \subseteq E$, $\mb{1}_S \in \mathbb{R}^E$ denotes the characteristic vector of $S$.
For a vector $\mb{x} \in \mathbb{R}^E$ and a set $S \subseteq E$, $\mb{x}(S)$ denotes the sum $\sum_{e \in S}\mb{x}(e)$. 

\subsection{Submodular Functions}

Let $F\colon 2^E \to \mathbb{R}_+$ be a set function.
We say that $F$ is \emph{monotone} if $F(S) \leq F(T)$ holds for every $S \subseteq T \subseteq E$.
We say that $F$ is \emph{submodular} if $F(S \cup \{e \}) - F(S) \geq F(T \cup \{e\}) - F(T)$ holds for any $S \subseteq T \subseteq E$ and $e \in E\setminus T$.

The \textit{multilinear extension} $f\colon{[0,1]}^E \to \mathbb{R}$ of a set function $F\colon 2^E \to \R$ is
$
  f(\mathbf{x})=\sum\limits_{S\subseteq E}F(S) \prod\limits_{e \in S}\x(e) \prod\limits_{e \not\in S}(1-\x(e)).
$

There is a probabilistic interpretation of the multilinear extension. Given $\mb{x} \in{[0, 1]}^E$ we can define $X$ to be the random subset of $E$ in which each element $e \in E$ is included independently with probability $\mb{x}(e)$ and is not included with probability $1-\mb{x}(e)$.
We write $X \sim \mb{x}$ to denote that $X$ is a random subset sampled this way from $\mb{x}$. Then we can simply write $f$ as $f(\mb{x})=\mathbb{E}_{X\sim \mb{x}}[F(X)]$.

Observe that for all $S\subseteq E$ we have $f(\mb{1}_S) = F(S)$. 
The following is well known:
\begin{proposition}[\citet{calinescu2011maximizing}]
  Let $f\colon {[0,1]}^E \to \mathbb{R}$ be the multilinear extension of a monotone submodular function $F\colon 2^E \to \mathbb{R}$. Then
  \begin{itemize}
      \item [1.] $f$ is monotone, meaning $\frac{\partial f}{\partial \mb{x}(e)}\geq 0$. Hence, $\nabla f(\mb{x})=(\frac{\partial f}{\partial \mb{x}(1)},\dots,\frac{\partial f}{\partial \mb{x}(n)})$ is a nonnegative vector. 
      \item[2.] $f$ is concave along any direction $\mathbf{d}\geq \mathbf{0}$.
      
  \end{itemize}
\end{proposition}

\subsection{$k$-submodular Functions}
 Given a natural number $k \geq 1$, a function $F\colon (k+
1)^E \to \mathbb{R}_+$ defined on $k$-tuples of pairwise disjoint subsets of $E$ is called \emph{k-submodular} if for all $k$-tuples $S = (S_1, \dots, S_k)$ and $T = (T_1, \dots , T_k)$ of pairwise disjoint subsets of $E$,
\[F(S)+F(T)\geq F(S\sqcap T)+F(S\sqcup T),\] where we define
\begin{align*}
S\sqcap T &= (S_1\cap T_1,\dots, S_k\cap T_k),\\
S\sqcup T &=
    \Bigg((S_1\cup T_1)\setminus \Bigg(\bigcup\limits_{i\neq 1}S_i\cup T_i \Bigg),\ldots, 
        (S_k\cup T_k)\setminus\Bigg(\bigcup\limits_{i\neq k}S_i\cup T_i \Bigg)\Bigg).
\end{align*}
\subsection{Matroids and Matroid Polytopes}

A pair $\mathcal{M} = (E,\mc{I})$ of a set $E$ and $\mc{I}\subseteq 2^E$ is called a \emph{matroid} if
\begin{inparaenum}
\item[1)] $\emptyset \in \mc{I}$,
\item[2)] $A\in \mc{I}$ for any $A \subseteq B\in \mc{I}$, and
\item[3)] for any $A,B\in \mc{I}$ with $|A| < |B|$, there exists $e \in B\setminus A$ such that $A\cup \{e\}\in \mc{I}$.
\end{inparaenum}
We call a set in $\mathcal{I}$ an \emph{independent set}. The \emph{rank function} $r_{\mc{M}}\colon 2^E \to \mathbb{Z}_+$ of $\mc{M}$ is
\[r_{\mc{M}}(S)=\max\{|I|:I\subseteq S, I\in \mc{I}\}.
\]
An independent set $S\in \mc{I}$ is called a \emph{base} if $r_{\mc{M}}(S)=r_{\mc{M}}(E)$. We denote the set of all bases by $\mc{B}$ and rank of $\mc{M}$ by $r(\mc{M})$.
The \emph{matroid polytope} $\mc{P}(\mc{M}) \subseteq \R^E$ of $\mc{M}$ is
$
  \mc{P}(\mc{M})=\mathrm{conv}\{\mb{1}_I : I\in \mc{I}\},
$
where $\mathrm{conv}$ denotes the convex hull.
Or equivalently~\cite{edmonds2003submodular},
\[
  \mc{P}(\mc{M})=\left\{\mb{x}\geq \mb{0} : \mb{x}(S)\leq r_{\mc{M}}(S) \; \forall S\subseteq E \right\}.
\]
Note that the matroid polytope is \textit{down-monotone}, that is, for any $\mb{x},\mb{y}\in \R^E$ with $\mb{0}\leq \mb{x}\leq \mb{y}$ and $\mb{y}\in \mc{P}(\mc{M})$ then $\mb{x}\in \mc{P}(\mc{M})$.

\begin{definition}[$\rho$-covering]
  Let $K \subseteq \mathbb{R}^E$ be a set.
  For $\rho > 0$, a set $C \subseteq K$ of points is called a \emph{$\rho$-covering} of $K$ if for any $\mathbf{x} \in K$, there exists $\mathbf{y} \in C$ such that $\|\mathbf{x}-\mathbf{y}\|\leq \rho$.
\end{definition}


\begin{theorem}[Theorem 5.5 of~\citet{yoshida2019cheeger}, paraphrased]
  \label{thm:covering}
  Let $\mathcal{M}=(E,\mathcal{I})$ be a matroid.
  For every $\epsilon > 0$, we can construct an $\epsilon B$-cover $C$ of $\mathcal{P}(\mathcal{M})$ of size $|E|^{O(1/\epsilon^2)}$ in $|E|^{O(1/\epsilon^2)}$ time, where $B$ is the maximum $\ell_2$-norm of a point in $\mathcal{P}(\mathcal{M})$.
\end{theorem}

\subsection{Differential Privacy}
The definition of differential privacy relies on the notion of neighboring datasets. Recall that two datasets are \textit{neighboring} if they differ in a single record. When two datasets $D, D'$ are neighboring, we write $D \sim D'$.
\begin{definition}[\citet{dwork2006our}]
  For $\epsilon,\delta \in \R_+$, we say that a randomized computation $M$ is \emph{$(\epsilon,\delta)$-differentially private} if for any neighboring datasets $D \sim D'$, and for any set of outcomes $S \subseteq \mathrm{range}(M)$,
  \[
    \Pr[M(D)\in S] \leq \exp(\epsilon) \Pr[M(D')\in S]+\delta.
  \]
  When $\delta=0$, we say $M$ is \emph{$\epsilon$-differentially private}.
\end{definition}

In our case, a dataset $D$ consists of \emph{private} submodular functions $F_1,\ldots,F_n\colon 2^E \to [0,1]$.
Two datasets $D$ and $D'$ are neighboring if all but one submodular function in those datasets are equal. The submodular function $F_D$ depends on the dataset $D$ in different ways, for example $F_D(S) = \sum\limits_{i=1}^{n} F_i(S)/n$ (CPP problem), or much more complicated ways than averaging functions associated to each individual.  

Differentially private algorithms must be calibrated to the sensitivity of the function of interest with respect to small changes in the input dataset, defined formally as follows.
\begin{definition}
  The sensitivity of a function $F_D\colon X \to Y$, parameterized by a dataset $D$, is defined as
   \[
    \max\limits_{D': D' \sim D} \max\limits_{x \in X} |F_D(x)-F_{D'}(x)|.
  \]
  A function with sensitivity $\Delta$ is called $\Delta$-sensitive.
\end{definition}

\subsubsection{Composition of Differential Privacy} 
  Let $\{(\epsilon_i, \delta_i)\}_{i=1}^k$ be a sequence of privacy parameters and let $M^*$ be a mechanism that behaves as follows on an input $D$. In each of rounds $i = 1, \dots, k$, the algorithm $M^*$ selects an $(\epsilon_i,\delta_i)$-differentially private algorithm $M_i$ possibly depending on the previous outcomes $M_1(D),\dots, M_i(D)$ (but not directly on the sensitive dataset $D$ itself), and releases $M_i(D)$. The output of $M^*$ is informally referred as the \emph{k-fold adaptive composition} of $(\epsilon_i,\delta_i)$-differentially private algorithms. For a formal treatment of adaptive composition, see~\citet{dwork2014algorithmic,dwork2010boosting}.
  We have the following guarantee on the differential privacy of the composite algorithm.
  \begin{theorem}\label{k-fold-composition}\cite{bun2016concentrated,dwork2009differential,dwork2010boosting}
  The $k$-fold adaptive composition of $k$ $(\epsilon_i, \delta_i)$-differentially private algorithms, with $\epsilon_i \leq \epsilon_0$ and $\delta_i \leq \delta_0$ for every $1\leq i\leq k$, satisfies $(\epsilon, \delta)$-differential privacy where
  \begin{itemize}
      \item $\epsilon=k\epsilon_0$ and $\delta=k\delta_0$ (the basic composition), or
      \item $\epsilon = \frac{1}{2}k\epsilon_0^2+ \sqrt{2\ln{1/\delta'}}\epsilon_0$ and $\delta =\delta' + k\delta$ for any $\delta' > 0$ (the advanced composition).
  \end{itemize}
\end{theorem}

\subsubsection{Exponential Mechanism}
One particularly general tool that we will use is the \emph{exponential mechanism} of~\citet{mcsherry2007mechanism}.
The exponential mechanism is defined in terms of a \textit{quality function} $q_D\colon \mathcal{R} \to \mathbb{R}$, which is parameterized by a dataset $D$ and maps a candidate result $R \in \mathcal{R}$ to a real-valued score.

\begin{definition}[\citet{mcsherry2007mechanism}]
  Let $\epsilon,\Delta > 0$ and let $q_D \colon \mathcal{R}\to \R$ be a quality score.
  Then, the \emph{exponential mechanism} ${EM}({\epsilon,\Delta,q_D})$  outputs $R \in \mathcal{R}$ with probability proportional to
  $
    \exp\left(\frac{\epsilon}{2\Delta}\cdot q_D(R)\right).
  $
\end{definition}

\begin{theorem}[\citet{mcsherry2007mechanism}]\label{thm:EM-bound}
  Suppose that the quality score $q_D \colon \mathcal{R}\to \R$ is $\Delta$-sensitive.
  Then, ${EM}({\epsilon,\Delta,q_D})$ is $\epsilon$-differentially private, and for every $\beta \in (0,1)$ outputs $R\in \mathcal{R}$ with
  \[
    \Pr\left[q_D(R) \geq \max_{R' \in \mathcal{R}} q_D(R')-\frac{2\Delta}{\epsilon}\ln\left(\frac{|\mathcal{R}|}{\beta}\right)\right] \geq 1-\beta.
  \]
\end{theorem}

%% file: continuous-greedy.tex
\section{Differentially Private Continuous Greedy Algorithm}
\label{section:DPCG}
In this section we prove Theorem~\ref{thm:main-introduction}. Throughout this section, we fix (private) monotone submodular functions $F_1,\ldots,F_n\colon 2^E \to [0,1]$, $\epsilon,\delta > 0$, and a matroid $M= (E,\mathcal{I})$.

Let $\mb{x}^* \in \mathcal{P}(\mc{M})$ be a maximizer of $f_D$.
We drop the subscript $D$ when it is clear from the context. Our algorithm (Algorithm~\ref{alg:DPGA}) is a modification of the continuous greedy algorithm~\cite{calinescu2011maximizing}.

\begin{algorithm}[tb]
  \caption{Differentially Private Continuous Greedy}
  \label{alg:DPGA}
  \begin{algorithmic}[1]
    \STATE {\bfseries Input: }{Submodular function $F_D\colon 2^E \to [0,1]$, dataset $D$, matroid $\mathcal{M}=(E,\mathcal{I})$, and $\epsprivacy > 0$ and $\rho \geq 0$.}
    \STATE Let $C_{\rho}$ be a $\rho$-covering of $\mathcal{P}(\mathcal{M})$, and $f_D$ be the multilinear extension of $F_D$.
    \STATE $\mb{x}_0 \gets \mb{0}$, $\epsscore \gets \frac{\epsprivacy}{2\Delta}$.
    \STATE $\alpha \gets \frac{1}{T}$, where $T = r(\mathcal{M})$.
    \FOR{$t=1$ to $T$}
    \STATE Sample $\mb{y}\in C_{\rho}$ with probability proportional to $\exp\bigl(\epsscore\ip{\mb{y}, \nabla f_D(\mb{x}_{t-1})}\bigr)$.
    \STATE Let $\mb{y}_{t-1}$ be the sampled vector.
    \STATE $\mb{x}_t\gets \mb{x}_{t-1}+\alpha \mb{y}_{t-1}$.
    \ENDFOR
    \STATE {\bfseries Output: } $\mb{x}_T$
  \end{algorithmic}
\end{algorithm}

\subsection{Approximation Guarantee}

\begin{lemma}\label{F-distance}
  For every $\mb{x},\mb{v} \in {[0,1]}^E$ with $\|\mb{v}\|_2\leq \rho$ and $\mb{x}+\mb{v} \in {[0,1]}^E$, we have $|f(\mb{x})-f(\mb{x}+\mb{v})|\leq 4\sqrt[4]{|E|}\sqrt{\rho}$.
\end{lemma}
\begin{lemma}\label{lem:existance-in-net}
  Suppose $\mb{y} \in {[0,1]}^E$ satisfies $\|\mb{y}-\mb{x}^*\|_2 \leq \rho$.
  Then for any $\mb{x} \in {[0,1]}^E$, we have $
    \ip{\mb{y}, \nabla f(\mb{x})} \geq f(\mb{x}^*)- f(\mb{x})- C_{\ref{lem:existance-in-net}}\sqrt{\rho}
  $
  for some constant $C_{\ref{lem:existance-in-net}} > 0$.
\end{lemma}
\begin{proof}
First, we show
  \[
    \ip{\mb{y}, \nabla f(\mb{x})} \geq f(\mb{y})-f(\mb{x}).
  \]
  Let us consider a direction $\mb{d} \in {[0,1]}^E$ such that $\mb{d}(e)=\max\{\mb{y}(e)-\mb{x}(e),0\}$ for every $e \in E$.
  Then, we have
  \begin{align*}
  \ip{\mb{y},\nabla f(\mb{x})}
  &\geq
  \ip{\mb{d},\nabla f(\mb{x})}\\
  &\geq
  f(\mb{x}+\mb{d})-f(\mb{x})\\
  &\geq f(\mb{y})-f(\mb{x}),
  \end{align*}
  where the first inequality follows from $\mb{y}\geq \mb{d}$ and $\nabla f (\mb{x})\geq 0$, the second inequality follows from the concavitity of $f$ along $\mb{d}$, and the third inequality follows from $\mb{x} + \mb{d} \geq \mb{y}$ and the monotonicity of $f$.  By Lemma~\ref{F-distance}, we have 
  \[f(\mb{y}) \geq f(\mb{x}^*)-4\sqrt[4]{|E|}\sqrt{\rho},\]
  which yields the desired result with $C_{\ref{lem:existance-in-net}} = 4\sqrt[4]{|E|}$.
\end{proof}


\begin{theorem}\label{thm:optimality}
   Suppose $F_D$ is $\Delta$-sensitive and $C_\rho$ is a $\rho$-covering of $\mc{P}(\mc{M})$. Then Algorithm~\ref{alg:DPGA}, with high probability, returns $\mb{x}_T\in\mc{P}(\mc{M})$ such that
  \[
  f_D(\mb{x}_T) \geq \left(1-\frac{1}{\mathrm{e}}\right)\mathrm{OPT} - O\left(C_{\ref{lem:existance-in-net}}\rho+\frac{\Delta r(\mc{M})\ln{|E|}}{\epsprivacy\rho^2}\right)
  \]
  Moreover, the algorithm evaluates $f_D$ at most $ O\left(r(\mathcal{M}) \cdot |C_{\rho}|\right)$ times.
\end{theorem}

\begin{proof}
  Clearly Algorithm~\ref{alg:DPGA} evaluates $f$ at most $O\left(r(\mc{M}) |C_{\rho}|\right)$ times.
  Observe that the algorithm forms a convex combination of $T$ vertices of the polytope $\mc{P}(\mc{M})$, each with weight $\alpha$  hence $\mb{x}_T\in\mc{P}(\mc{M})$.
  In what follows, we focus on the quality of the output of the algorithm. Suppose $\mb{y}'\in C_{\rho}$ with $\|\mb{y}'-\mb{x}^*\|_2 \leq \rho$.
  By Theorem~\ref{thm:EM-bound}, with probability at least $1-\frac{1}{|E|^2}$, we have
  \begin{align*}
    \ip{\mb{y}_{t}, \nabla f(\mb{x}_{t})}
    &\geq \argmax_{\mb{y}\in C_{\rho}}\ip{\mb{y}, \nabla f(\mb{x}_{t})} - \frac{2 \Delta}{\epsprivacy}\ln (|E|^2|C_{\rho}|)\\
    & \geq \ip{\mb{y}', \nabla f(\mb{x}_{t})} - \frac{2\Delta}{\epsprivacy}\ln (|E|^2|C_{\rho}|)\\
    & \overset{\text{By Lemma~\ref{lem:existance-in-net}}}{\geq}
    \begin{multlined}[t]
    f(\mb{x}^*) - f(\mb{x}_{t})  -  C_{\ref{lem:existance-in-net}}\sqrt{\rho}  - \frac{2 \Delta}{\epsprivacy}\ln (|E|^2|C_{\rho}|)
    \end{multlined}
  \end{align*}
  By a union bound, with probability at least $1-\frac{1}{poly(|E|)}$, the above inequality holds for every $t$.
  In what follows, we assume this has happened. Further, let us assume that $t$ is a continuous variable in $[0,T]$.
  We remark that discretization of $t$ in our algorithm introduces error into the approximation guarantee.
  However, this can be handled by sufficiently large $T$, say, $r(\mathcal{M})$ as in Algorithm~\ref{alg:DPGA}, and small step size $\alpha$~\cite{calinescu2011maximizing}.
  In what follows $t$ is assumed to be continuous and we write $\frac{d\mb{x}_t}{dt}=\alpha\mb{y}_t$, hence
  \begin{align*}
      \frac{d f(\mb{x}_t)}{dt}
      &=\sum_{e}\frac{\partial f(\mb{x}_t(e))}{\partial \mb{x}_t(e)}\frac{d\mb{x}_t(e)}{dt}\\
      &= \nabla f(\mb{x}_t)\cdot \frac{d\mb{x}_t}{dt}
      =\alpha\ip{\mb{y}_t,\nabla f(\mb{x}_t)}\\
      & \geq
      \alpha\bigg(f(\mb{x}^*) - f(\mb{x}_{t}) - C_{\ref{lem:existance-in-net}}\sqrt{\rho}  - \frac{2\Delta}{\epsprivacy}\ln (|E|^2|C_{\rho}|)\bigg),
  \end{align*}
  where the first equality follows from the chain rule.
  Let $\beta=f(\mb{x}^*) - C_{\ref{lem:existance-in-net}}\sqrt{\rho}  - \frac{ 2\Delta}{\epsprivacy}\ln (|E|^2|C_{\rho}|)$. Solving the following differential equation $\frac{d f(\mb{x}_t)}{dt}=\alpha(\beta - f(\mb{x}_{t}))$ with $f(\mb{x}_0)=0$ gives us $f(\mb{x}_t) = \beta(1-\mathrm{e}^{-\alpha t}).$ For $\alpha = \frac{1}{T}, t=T$ we obtain
  \begin{align*}
       f(\mb{x}_T)
      &= \beta(1-\mathrm{e}^{-1})\\
      &=
      \left(1-\frac{1}{\mathrm{e}}\right) f(\mb{x}^*) - O\left(C_{\ref{lem:existance-in-net}}\sqrt{\rho}  + \frac{2\Delta}{\epsprivacy}\ln ( |E|^2|C_{\rho}|)\right)
      \\
      &=
      \left(1-\frac{1}{\mathrm{e}}\right) f(\mb{x}^*) - O\left(C_{\ref{lem:existance-in-net}}\sqrt{\rho} + \frac{\Delta}{\epsprivacy}(\ln |E|+\ln |E|^{{\left(\frac{B}{\rho}\right)}^2})\right) 
      \\
      &=
      \left(1-\frac{1}{\mathrm{e}}\right) f(\mb{x}^*) - O\left(C_{\ref{lem:existance-in-net}}\sqrt{\rho} + \frac{\Delta}{\epsprivacy}{\left(\frac{B}{\rho}\right)}^2 \ln |E|\right)\tag{$B^2\leq r(\mc{M})$}
      \\
      &
      \geq \left(1-\frac{1}{\mathrm{e}}\right) f(\mb{x}^*) - O\left(C_{\ref{lem:existance-in-net}}\sqrt{\rho} + \frac{\Delta r(\mc{M})\ln|E|}{\epsprivacy\rho^2}\right)
      \qedhere
  \end{align*}
\end{proof}

\begin{remark}\label{remark1}
As already pointed out in the proof of Theorem~\ref{thm:optimality}, the discretization of $t$ introduces error into the approximation guarantee yielding $(1-1/\mathrm{e}-1/\mathrm{poly}(|E|))\mathrm{OPT}$.
However, this can be shaved off to $(1-1/\mathrm{e})\mathrm{OPT}$ by sufficiently large $T$~\cite{calinescu2011maximizing}.
Moreover, evaluating $f$ (even approximately) is expensive.
To achieve the nearly optimal approximation guarantees, the evaluation error needs to be very small and in a lot of cases, the error needs to be $O(1/|E|)$ times the function value.
As a result, a single evaluation of the multilinear extension $f$ requires $\Omega(|E|)$ evaluations of $F$ (see~\citet{Ene2018} for recent improvement). Therefore, our algorithm requires $O(r(\mc{M})|E||C_{\rho}|)$ evaluation of $F$.
\end{remark}
\begin{remark}\label{remark2}
From a fractional solution $\mb{x}^*$, we can obtain an integral solution $\mb{s} \in \{0,1\}^E$ such that $f(\mb{s}) \geq f(\mb{x}^*)$.
Such an integer solution corresponds to a vertex of $\mc{P}(\mc{M})$ and hence a discrete solution $S \in \mc{I}$. This can be done using the so-called \emph{swap rounding}~\cite{ChekuriVZ10}.
\end{remark}

\subsection{Privacy Analysis}

\begin{theorem}\label{thm:privacy}
  Algorithm~\ref{alg:DPGA} preserves $O(\epsprivacy r(\mathcal{M})^2)$-differential privacy.
\end{theorem}

\begin{proof}
    Let $D$ and $D'$ be two neighboring datasets and $F_D,F_{D'}$ be their associated functions. For a fixed $\mb{y}_t\in C_{\rho}$, we consider the relative probability of Algorithm~\ref{alg:DPGA} (denoted by $M$) choosing $\mb{y}_t$ at time step $t$ given multilinear extensions of $F_D$ and $F_{D'}$. Let $M_t(f_D \mid \mb{x}_t)$ denote the output of $M$ at time step $t$ given dataset $D$ and point $\mb{x}_t$. Similarly, $M_t(f_{D'} \mid \mb{x}_t)$ denotes the output of $M$ at time step $t$ given dataset $D'$ and point $\mb{x}_t$. Further, write $d_{\mb{y}}=\ip{\mb{y},\nabla f_D(\mb{x}_t)}$ and  $d'_{\mb{y}}=\ip{\mb{y},\nabla f_{D'}(\mb{x}_t)}$. We have 
  \begin{flalign*}
      \frac{\Pr[M_t(f_D \mid \mb{x}_t)
      =\mb{y}_t]}{\Pr[M_t(f_{D'} \mid \mb{x}_t)=\mb{y}_t]}
      &=
       \frac{\exp(\epsscore\cdot d_{\mb{y}_t})/\sum_{\mb{y} \in C_\rho}\exp(\epsscore\cdot d_{\mb{y}})}{\exp(\epsscore\cdot d'_{\mb{y}_t})/\sum_{\mb{y} \in C_\rho}\exp(\epsscore\cdot d'_{\mb{y}})}\\
     & = \frac{\exp(\epsscore\cdot d_{\mb{y}_t})}{\exp(\epsscore\cdot d'_{\mb{y}_t})} \cdot \frac{\sum_{\mb{y} \in C_\rho}\exp(\epsscore\cdot d'_{\mb{y}})}{\sum_{\mb{y} \in C_\rho}\exp(\epsscore\cdot d_{\mb{y}})}.
  \end{flalign*}
  For the first factor, we have
  \begin{align*}
     \frac{\exp(\epsscore\cdot d_{\mb{y}_t})}{\exp(\epsscore\cdot d'_{\mb{y}_t})}
    &= \exp\bigl(\epsscore( d_{\mb{y}_t}-d'_{\mb{y}_t})\bigr)\\
    &=\exp\bigl(\epsscore(\ip{\mb{y}_t,\nabla f_D(\mb{x}_t) - \nabla f_{D'}(\mb{x}_t) })\bigr)\\
    &\leq \exp\bigl(\epsscore\|\mb{y}_t\|_1 \|\nabla f_D(\mb{x}_t) - \nabla f_{D'}(\mb{x}_t)\|_\infty\bigr)\\
    &=\exp\Bigg(
                \epsscore \sum_{e \in E}\mb{y}_t(e) \cdot
                \bigg(
                    \max_{e \in E}\mathop{\mathbb{E}}_{R \sim  \mb{x}_t}\Big[
                    F_D(R \cup
                    \{e \})
                -F_D(R)
                -F_{D'}(R \cup \{e \})+F_{D'}(R)\Big] \bigg)
            \Bigg)\\
    & \leq \exp(O(\epsscore\cdot r(\mathcal{M})\cdot 2\Delta))
    =\exp(O(\epsprivacy\cdot r(\mathcal{M})))
  \end{align*}
  Note that the last inequality holds since $\mb{y}_t$ is a member of the matroid polytope $\mc{P}(\mc{M})$ and by definition we have $\sum_{e \in E}\mb{y}_t(e)\leq r_{\mc{M}}(E)=r(\mc{M})$. Moreover, recall that $F_D$ is $\Delta$-sensitive. 

  For the second factor, let us write $\beta_{\mb{y}}=d'_{\mb{y}}-d_{\mb{y}}$ to be the \textit{deficit} of the probabilities of choosing direction $\mb{y}$ in instances $f_{D'}$ and $f_D$. Then, we have
  \begin{align*}
     \frac{\sum_{\mb{y} \in C_\rho}\exp(\epsscore\cdot d'_{\mb{y}})}{\sum_{\mb{y} \in C_\rho}\exp(\epsilon'\cdot d_{\mb{y}})}
    &=\frac{\sum_{\mb{y} \in C_\rho}\exp(\epsscore\cdot \beta_{\mb{y}})\exp(\epsscore\cdot d_{\mb{y}})}{\sum_{\mb{y} \in C_\rho}\exp(\epsscore\cdot d_{\mb{y}})}\\
    &=\mathbb{E}_{\mb{y}}[\exp(\epsscore\cdot \beta_{\mb{y}})]
    \leq \exp\bigl(O(\epsscore\cdot r(\mathcal{M})\cdot 2\Delta)\bigr)\\
    &=\exp\bigl(O(\epsprivacy\cdot r(\mathcal{M}))\bigr). 
  \end{align*}
  The expectation is taken over the probability distribution over $\mb{y}$ selected at time $t$ in instance with input $D$. Recall that we choose $\mb{y}$ with probability proportional to $\exp(\epsilon' d_{\mb{y}})$. By a union bound, Algorithm~\ref{alg:DPGA} preserves $O({\epsprivacy T r(\mathcal{M})})\leq O(\epsprivacy r(\mathcal{M})^2)$-differential privacy. To obtain an integral solution from a fractional solution, we use swap rounding technique (see Remark~\ref{remark2}) which does not depend on the input function and hence preserves the privacy.
\end{proof}

 Note that the privacy factor in the work of \citet{MitrovicB0K17} is $O(\epsprivacy r(\mc{M}))$. However, our privacy factor is $O(\epsprivacy r(\mc{M})^2)$, this is because we deal with the multilinear extension of a submodular function rather than the function itself (which is different from the previous works). 

 \begin{theorem}[Formal version of Theorem~\ref{thm:main-introduction}]
 \label{thm:main-introduction-formal}
    Suppose $F_D$ is $\Delta$-sensitive and Algorithm~\ref{alg:DPGA} is instantiated with 
    $\rho=\frac{\epsprivacy}{|E|^{1/2}}$. Then Algorithm~\ref{alg:DPGA} is $(\epsprivacy r(\mc{M})^2)$-differentially private and, with high probability, returns $S\in \mathcal{I}$ with quality at least
    \[
    F_D(S) \geq \left(1-\frac{1}{\mathrm{e}}\right)\mathrm{OPT} -O\left(\sqrt{\epsprivacy}+\frac{\Delta r(\mc{M})|E|\ln{|E|}}{\epsprivacy^3}\right) 
    \]
\end{theorem}

\begin{example}[Maximum Coverage]
Let $G = (U, V, E)$ be a bipartite graph, and $B$ be a budget constraint. In Maximum Coverage problem, the goal is to find a set $S$ of ${B}$ vertices in $U$ so that the number of vertices in $V$ incident to some vertex in $S$ is maximized. The edges incident to a vertex $v \in V$ are private information about $v$. If we instantiate Theorem~\ref{thm:main-introduction-formal} on this problem, the privacy factor is $\epsilon B^2$ and the additive error is $O(\Delta B|U| \ln(|U|)/ \epsilon^3)$, where $\Delta$ is the maximum degree of a vertex in $V$. To have a meaningful privacy bound, we set $\epsprivacy \ll 1/B^2$, and the additive error becomes $\Delta B^7|U| \ln(|U|)$. However, OPT could be $\Omega(|V|)$, which is much larger than the additive error when $|V| \gg |U|$.  Indeed, by optimizing $\rho$, we can improve the additive error to $O(\Delta B^3 |U| \ln(|U|))$, which will be more practical. 
\end{example}


%% file: k-submodular.tex

\section{$k$-Submodular Function Maximization}\label{section;k-submodular}

In this section, we study a natural generalization of submodular
functions, namely $k$-submodular functions. 
Associate $(S_1, \dots, S_k)\in {(k+1)}^E$ with $\mb{s}\in{\{0,1,\dots,k\}}^E$ by $S_i=\{e\in E\mid\mb{s}(e)=i\}$ for $i\in[k]$ and define the \emph{support} of $\mb{s}$ as $\mathrm{supp}(\mb{s})=\{e\in E\mid \mb{s}(e)\neq 0\}$. Let $\preceq$ be a partial ordering on ${(k+1)}^E$ such that, for $\mb{s}=(S_1,\dots,S_k)$ and $\mb{t}=(T_1, \dots , T_k)$ in ${(k+1)}^E$, $\mb{s}\preceq \mb{t}$ if $S_i\subseteq T_i$ for every $i\in [k]$. We say that a function $F\colon {(k+1)}^E \to \mathbb{R}_+$ is \emph{monotone} if $F(\mb{s})\leq F(\mb{t})$ holds for every $\mb{s}\preceq \mb{t}$.
Define the \emph{marginal gain} of adding $e\not\in\bigcup_{\ell\in[k]}S_{\ell}$ to the $i$-th set of $\mb{s}\in {(k+1)}^E$ to be
\begin{align*}
    \Delta_{e,i}F(\mb{s})
    =F(S_1,\ldots,S_{i-1},S_i\cup\{e\},S_{i+1},\dots,S_k)
    -F(S_1,\dots,S_k).
\end{align*}
The monotonicity of $F$ is equivalent to $\Delta_{e,i}F(\mb{s}) \geq 0$ for any $\mb{s} = (S_1,\dots,S_k)$ and $e\not\in\bigcup_{\ell\in[k]}S_{\ell}$ and $i \in [k]$.

Our goal is maximizing a monotone $k$-submodular function under matroid constraints. That is, given a monotone $k$-submodular function $F_D\colon {(k+1)}^E \to \mathbb{R}_+$ and a matroid $\mc{M}=(E,\mc{I})$, we want to solve the following problem.
\begin{align*}
  \max\limits_{\mb{x}\in {(k+1)}^E} F_D(\mb{x})
  \quad \text{subject to } \bigcup_{i\in[k]} X_i \in \mc{I}
\end{align*}
The following are known due to~\citet{sakaue2017maximizing}. They may have appeared in other literature that we are not aware of.
\begin{lemma}[\citet{sakaue2017maximizing}]\label{Lem-opt-size}
For any maximal optimal solution $\mb{o}$ we have $|\mathrm{supp}(\mb{o})|= r(\mc{M})$.
\end{lemma}

\begin{lemma}[\citet{sakaue2017maximizing}]\label{extension}
Suppose $A \in \mc{I}$ and $B \in \mc{B}$ (recall $\mc{B}$ denotes the set of bases) satisfy $A \subseteq B$. Then, for any $e \not\in A$ satisfying $A \cup \{e\} \in \mc{I}$, there exists $e' \in B \setminus A$ such that $B \setminus\{e'\} \cup \{e\} \in \mc{B}$.
\end{lemma}
\begin{algorithm}[tb]
  \caption{Differentially private $k$-submodular maximization with a matroid constraint}
  \label{alg:DP-k-submodular-matroid}
  \begin{algorithmic}[1]
    \STATE {\bfseries Input:} {monotone $k$-submodular functions $F_D\colon {(k+1)}^E \to [0,1]$, a matroid $\mathcal{M}=(E,\mathcal{I})$, and $\epsprivacy > 0$.}
    \STATE $\mb{x} \gets \mb{0}$, $\epsscore \gets \frac{\epsprivacy}{2\Delta}$
    \FOR{$t=1$ to $r(\mc{M})$}
    \STATE Let $\Lambda(\mb{x})=\{e\in E\setminus \mathrm{supp}(\mb{x})\mid \mathrm{supp}(\mb{x})\cup\{e\}\in \mc{I}\}$
    \STATE Choose $e\in \Lambda(\mb{x})$ and $i\in[k]$ with probability proportional to $\exp(\epsscore\Delta_{e,i}F_D(\mb{x}))$.
    \STATE $\mb{x}(e)\gets i$.
    \ENDFOR
    \STATE {\bfseries Output:} $\mb{x}$
  \end{algorithmic}
\end{algorithm}
Having Lemma~\ref{Lem-opt-size}, our algorithm runs in $r(\mc{M})$ iterations and at each iteration chooses an element $e$ with probability proportional to $\exp(\epsscore\Delta_{e,i}F_D(\mb{x}))$ and adds $e$ to $\mathrm{supp}(\mb{x})$. The analysis for the approximation guarantee is similar to the ones in~\citet{iwata2016improved,ohsaka2015monotone,sakaue2017maximizing,ward2014maximizing} and relies on Theorem~\ref{thm:EM-bound}.  
\begin{theorem}\label{opt-k-submodular}
Suppose $F_D$ has sensitivity $\Delta$. Then Algorithm~\ref{alg:DP-k-submodular-matroid}, with high probability, returns $\mb{x} \in {(k+1)}^E$ such that $\mathrm{supp}(\mb{x}) \in \mc{B}$ and $F_D(\mb{x}) \geq \frac{1}{2}\mathrm{OPT}-O(\frac{\Delta r(\mc{M})\ln{|E|}}{\epsprivacy})$.
\end{theorem}
The privacy guarantee follows immediately from the $\epsprivacy$-differential privacy of the exponential mechanism, together with Theorem~\ref{k-fold-composition}.
\begin{theorem}\label{privacy-k-submodular}
   Algorithm~\ref{alg:DP-k-submodular-matroid} preserves $O(\epsprivacy r(\mc{M}))$-differential privacy.
  It also provides $(\frac{1}{2}r(\mc{M})\epsprivacy^2+ \sqrt{2\ln{1/\delta'}}\epsprivacy,\delta')$-differential privacy for every $\delta' > 0$.
\end{theorem}
Clearly, Algorithm~\ref{alg:DP-k-submodular-matroid} evaluates $F_D$ at most $O(k|E|r(\mc{M}))$ times. Next theorem summarizes the results of this section.
\begin{theorem}
   Suppose $F_D$ has sensitivity $\Delta$. Then Algorithm~\ref{alg:DP-k-submodular-matroid}, with high probability, outputs a solution $\mb{x} \in {(k+1)}^E$ such that $\mathrm{supp}(\mb{x})$ is a base of $\mc{M}$ and $F_D(\mb{x}) \geq \frac{1}{2}\mathrm{OPT}-O(\frac{\Delta r(\mc{M})\ln{|E|}}{\epsprivacy})$ by evaluating $F_D$ at most  $O(k|E|r(\mc{M}))$ times. Moreover, this algorithm preserves $O(r(\mc{M})\epsprivacy)$-differential privacy.
\end{theorem}

\subsection{Improving the Query Complexity}
By applying a sampling technique~\cite{MirzasoleimanBK15,ohsaka2015monotone}, we improve the number of evaluations of $F$ from $O(k|E|r(\mc{M}))$ to $O(k|E|$ $\ln{r(\mc{M})} \ln{\frac{r(\mc{M})}{\gamma}})$,
where $\gamma > 0$ is a failure probability. Hence, even when $r(\mc{M})$ is as large as $|E|$, the number of function evaluations is almost linear in $|E|$. The main difference from Algorithm~\ref{alg:DP-k-submodular-matroid} is that we sample a
sufficiently large subset $R$ of $E$, and then greedily assign a value only looking at elements in $R$. 
\begin{algorithm}[t!]
  \caption{Improved differentially private $k$-submodular maximization with a matroid constraint}\label{alg:IDP-k-submodular-matroid}
  \begin{algorithmic}[1]
    \STATE {\bfseries Input: }{monotone $k$-submodular functions $F_D\colon {(k+1)}^E \to [0,1]$, a matroid $\mathcal{M}=(E,\mathcal{I})$, $\epsprivacy > 0$, and a failure probability $\gamma>0$.}
    \STATE $\mb{x} \gets \mb{0}$, $\epsscore \gets \frac{\epsprivacy}{2\Delta}$
    \FOR{$t=1$ to $r(\mc{M})$}
    \STATE $R\gets$ a random subset of size $\min\{\frac{|E|-t+1}{r(\mc{M})-t+1} \log\frac{r(\mc{M})}{\gamma} , |E|\}$ uniformly sampled from $E\setminus \mathrm{supp}(\mb{x})$.

    \STATE Choose $e\in R$ with $\mathrm{supp}(\mb{x})\cup\{e\}\in \mc{I}$ and $i\in[k]$  with probability proportional to $\exp(\epsscore\Delta_{e,i}F_D(\mb{x}))$.
    \STATE $\mb{x}(e)\gets i$.
    \ENDFOR
    \STATE {\bfseries Output:} $\mb{x}$
  \end{algorithmic}
\end{algorithm}

\begin{theorem}\label{Imp-k-sub-opt}
  Suppose $F_D$ has sensitivity $\Delta$. Then  Algorithm~\ref{alg:IDP-k-submodular-matroid}, with probability at least $1-\gamma$, outputs a  solution with quality at least $\frac{1}{2}\mathrm{OPT}-O\left(\frac{\Delta r(\mc{M})\ln{\frac{|E|}{\gamma}}}{\epsprivacy}\right)$ by evaluating $F_D$ at most $O\left(k|E|\ln{r(\mc{M})} \ln{\frac{r(\mc{M})}{\gamma}} \right)$ times.
\end{theorem}
Similar to Theorem~\ref{privacy-k-submodular} and using the composition Theorem~\ref{k-fold-composition}, Algorithm~\ref{alg:IDP-k-submodular-matroid} preserves $O(\epsprivacy r(\mc{M}))$-differential privacy. It also provides $O\left(\frac{1}{2}r(\mc{M})\epsprivacy^2+\sqrt{2\ln{1/\delta'}}\epsprivacy,\delta'\right)$-differential privacy for every $\delta'>0$. In summary, we have

\begin{theorem}
    Suppose $F_D$ has sensitivity $\Delta$. Then, with probability at least $1-\gamma$, Algorithm~\ref{alg:IDP-k-submodular-matroid} returns a solution $\mb{x} \in {(k+1)}^E$ such that $\mathrm{supp}(\mb{x}) \in \mc{B}$ and $F_D(\mb{x}) \geq \frac{1}{2}\mathrm{OPT}-O\left(\frac{\Delta r(\mc{M})\ln{\frac{|E|}{\gamma}}}{\epsprivacy}\right)$ by evaluating $F_D$ at most $O\left(k|E|\ln{r(\mc{M})} \ln{\frac{r(\mc{M})}{\gamma}} \right)$ times. Moreover, this algorithm preserves $O(\epsprivacy r(\mc{M}))$-differential privacy.
\end{theorem}

\subsection{Motivating Examples}\label{section:k-sub-examples}
\begin{example}
Suppose that we have $m$ ad slots and $k$ ad agencies, and we want to allocate at most $B (\leq m)$ slots to the ad agencies. Each ad agency $i$ has a “influence graph” $G_i$, which is a bipartite graph $(U, V, E_i)$, where $U$ and $V$ correspond to ad slots and users, respectively, and an edge $uv \in E_i$ indicates that if the ad agency $i$ takes the ad slot $u$ (and put an ad there), the user $v$ will be influenced by the ad. The goal is to maximize the number of influenced people (each person will be counted multiple times if he/she is influenced by multiple ad agencies), based on which we get revenue from the ad agencies. This problem can be modeled as $k$-submodular function maximization under a cardinality constraint (a special case of matroid constraints), and edges incident to a user $v$ in $G_1,\dots,G_k$ are sensitive data about $v$. 
\end{example}
\begin{example}
Another example comes from (a variant of) facility location. Suppose that we have a set $E$ of $n$ lands, and we want to provide $k$ 
resources (e.g., gas and electricity) to all the lands by opening up 
facilities at some of the lands. For each resource type $i$ and lands $e, e' \in E$, we have a cost $c_i(e, e')$ of sending the resource of type $i$ from $e$ to $e'$. For a set $S \subseteq E$, let $c_i(e, S) 
= \min_{e' \in S} c_i(e, e')$, which is the cost of sending a 
resource of type $i$ to e when we open up facilities of type $i$ at 
lands in $S$. Assume we cannot open two or more facilities in the 
same land. Then, the goal is to find disjoint sets $S_1,\dots, S_k$ 
with $\sum_i |S_i| <= B$ for some fixed B that maximize $\sum_e 
\sum_i (C - c_i(e, S_i))$, where $C$ is a large number so that the 
objective function is always non-negative. This problem can be 
modeled as $k$-submodular function maximization under a cardinality 
constraint, and the costs $c_i(e, \cdot)$ are sensitive data about 
$e$.
\end{example}

%% file: conclusion.tex
\vspace{-1em}
\section{Conclusion}
We proposed a differentially private algorithm for maximizing monotone submodular functions under matroid constraint. Our algorithm provides the best possible approximation guarantee that matches the approximation guarantee in non-private setting. It also has a competitive number of function evaluations that is significantly faster than the non-private one. We also presented a differentially private algorithm for $k$-submodular maximization under matroid constraint that uses almost liner number of function evaluations and has an asymptotically tight approximation ratio.

%% file: appendix.tex
\section{Probability Distributions}
Let $P$ be a probability distribution over a finite set $E$.
For an element $e \in E$, we write $P(e)$ to denote the probability that $e$ is sampled from $P$.

Let $P$ and $Q$ be two distributions over the same set $E$.
The \textit{total variation distance} and the \emph{Hellinger distance} between $P$ and $Q$ are
\begin{align*}
  d_{\mathrm{TV}}(P,Q) = \frac{1}{2}\sum\limits_{e \in E} |P(e)-Q(e)|
  \quad \text{and} \quad \\
  h(P,Q)=\frac{1}{\sqrt{2}}\sqrt{\sum_{e \in E}{\left(\sqrt{P(e)}-\sqrt{Q(e)}\right)}^2},
\end{align*}
respectively.
It is well known that $d_{\mathrm{TV}}(P,Q) \leq \sqrt{2}h(P,Q)$ holds.

For two distributions $P$ and $Q$, we denote by $P \otimes Q$ their product distribution.
The following is well known:
\begin{lemma}\label{lem:hellinger-product-distribution}
  Let $P_1,\ldots,P_n$ and $Q_1,\ldots,Q_n$ be probability distributions over $E$.
  Then, we have 
  \begin{align*}
    {h(P_1\otimes P_2\otimes \cdots \otimes P_n, Q_1\otimes Q_2\otimes \cdots \otimes Q_n)}^2
    \leq \sum\limits_{i=1}^{n}h(P_i,Q_i)^2.
  \end{align*}
\end{lemma}
Finally, we use the following result due to Hoeffding in order to bound the error of our sampling step in Section~\ref{section:IDPCG}. 
\begin{theorem}[Hoeffding's inequality~\cite{hoeffding1963probability}]
\label{Hoeffding}
  Let $X_1, \ldots, X_n$ be independent random variables bounded by the interval $[0, 1]: 0 \leq X_i \leq 1$. We define the empirical mean of these variables by $\bar{X}=\frac{1}{n}(X_1+\cdots+X_n)$. Then
  \[
    \Pr[\bar{X}-E[\bar{X}]\geq t] \leq \exp(-2nt^2).
  \]
\end{theorem}

\section{Missing Proofs from Section \ref{section:DPCG}}

\begin{proof}[Proof of Lemma~\ref{F-distance}]
We have
  \begin{align}
     |f(\mb{x})-f(\mb{x}+\mb{v})| \nonumber 
    &= 
    \Bigg|\sum\limits_{S\subseteq E}F(S) \bigg(\prod\limits_{e\in S}\mb{x}(e) \prod\limits_{e\not\in S}\bigl(1-\mb{x}(e)\bigr) 
    -\prod\limits_{e\in S}\bigl(\mb{x}(e)+\mb{v}(e)\bigr) \prod\limits_{e\not\in S}\bigl(1-\mb{x}(e)-\mb{v}(e)\bigr)\bigg)\Bigg| \nonumber 
    \\
    &\leq
     \sum\limits_{S\subseteq E} \Bigg|\prod\limits_{e\in S}\mb{x}(e) \prod\limits_{e\not\in S}\bigl(1-\mb{x}(e)\bigr)
    - \prod\limits_{e\in S}\bigl(\mb{x}(e)+\mb{v}(e)\bigr) \prod\limits_{e\not\in S}(1-\mb{x}(e)-\mb{v}(e))\Bigg|. \label{eq:F-distance-1}
  \end{align}
  Now, we define probability distributions ${\{ P_e\}}_{e \in E}$ and ${\{Q_e\}}_{e \in E}$ over $\{0,1 \}$ so that $P_e(1) = \mb{x}(e)$ and $Q_e(1) = \mb{x}(e)+\mb{v}(e)$, respectively, for every $e \in E$.
  Note that
  \begin{align*}
     g(\mb{x}(e))&={h(P_e,Q_e)}^2\\
     &= {\left(\sqrt{\mb{x}(e)}-\sqrt{\mb{x}(e)+\mb{v}(e)}\right)}^2 
     +{\left(\sqrt{1-\mb{x}(e)}-\sqrt{1-\mb{x}(e)-\mb{v}(e)}\right)}^2 
  \end{align*}
  is a convex function with domain $\mb{x}(e)\in [0,1-\mb{v}(e)]$. The maximum value for this function happens at $\mb{x}(e)=0$ and $\mb{x}(e)=1-\mb{v}(e)$. Further its minimum is at $\mb{x}(e)=[1-\mb{v}(e)] / 2$. 
  \begin{align*}
      {h(P_e,Q_e)}^2 
      &= g(\mb{x}(e)) \\
      &\leq g(0)\\
      & =g(1-\mb{v}(e))\\
      & = 2-2\sqrt{1-\mb{v}(e)}\\
      & \leq {\mb{v}(e)}^2 + \mb{v}(e) \\
      & \leq 2\mb{v}(e) \tag{for $\mb{v}(e)\in [0,1]$}
  \end{align*}

  Letting $P = \bigotimes_{e \in E}P_e$ and $Q = \bigotimes_{e \in E}Q_e$, we have
  \begin{align*}
   \eqref{eq:F-distance-1} 
  & \leq 2 \cdot d_{\mathrm{TV}}(P,Q) \\
  &= 2\sqrt{2} \cdot h(P,Q) \\
  & \leq 2\sqrt{2}\sqrt{\sum\limits_{e \in E}{h(P_e,Q_e)}^2}  \tag{By Lemma~\ref{lem:hellinger-product-distribution}}\\
  & = 2\sqrt{2}\sqrt{\sum\limits_{e\in E}2{\mb{v}(e)}} \\
  & = 4\sqrt{|\mb{v}|_1} \\
  & \leq 4 \sqrt{\sqrt{|E|} \|\mb{v}\|_2} \\
  & \leq 4\sqrt[4]{|E|}\sqrt{\rho}  
  \qedhere
  \end{align*}
\end{proof}

\section{Missing Proofs from Section \ref{section:IDPCG}}
\subsection{Proof of Lemma~\ref{lemma:EM-Analogous}}
\begin{proof}[Proof of Lemma~\ref{lemma:EM-Analogous}]
 Let $OPT=\max_{\mb{z} \in C_{\rho}} \ip{\mb{z},\nabla f(\mb{x}_{t-1})}$ and $q_t(\mb{z})=\ip{\mb{z},\nabla f(\mb{x}_{t-1})}$ for every $\mb{z}\in \mc{P}(\mc{M})$. Further, let $\mb{y}_t$ be the output of the algorithm and $\tilde{L}^{\mb{x}_{t-1}}_{\mu,i}$ denote the estimated size of the $i$-th layer. 
 \begin{align*}
    \Pr\left[q(\mb{y}_t) \leq OPT -\frac{2\Delta}{\epsprivacy}\xi\right] 
    &\leq \frac{\Pr[q(\mb{y}_t) \geq OPT -\frac{2\Delta}{\epsprivacy}\xi]}{\Pr[q(\mb{y}_t) = OPT ]}\\
    & \leq 
    \frac{\exp\left[\epsscore\left( OPT-\frac{2\Delta}{\epsprivacy}\xi+\ln(1+\mu)\right)\right]}{\sum\limits_{j=1}^{k}\tilde{L}^{\mb{x}_{t-1}}_{\mu,j}(1+\mu)^{\epsscore (j-1)}}
    \times
    \frac{\sum\limits_{j=1}^{k}|\mc{L}^{\mb{x}_{t-1}}_{\mu,j}|(1+\mu)^{\epsscore (j-1)}}{\exp(\epsscore OPT)}
    \\
    &= 
    \frac{\exp\left[\epsscore\left( OPT-\frac{2\Delta}{\epsprivacy}\xi+\ln(1+\mu)\right)\right]}{\exp(\epsscore OPT)}
    \times 
    \frac{\sum\limits_{j=1}^{k}|\mc{L}^{\mb{x}_{t-1}}_{\mu,j}|(1+\mu)^{\epsscore (j-1)}}{\sum\limits_{j=1}^{k}\tilde{L}^{\mb{x}_{t-1}}_{\mu,j}(1+\mu)^{\epsscore (j-1)}}
\end{align*}
Consider the first term,
\begin{align*}
    \frac{\exp\left[\epsscore\left( OPT-\frac{2\Delta}{\epsprivacy}\xi+\ln(1+\mu)\right)\right]}{\exp(\epsscore OPT)}
    &=\exp\left[\epsscore\left(-\frac{2\Delta}{\epsprivacy}\xi+\ln(1+\mu)\right)\right]\\
    &= \exp(-\xi)\exp\left(\epsscore\ln(1+\mu)\right)\\
    &=\exp(-\xi)(1+\mu)^{\epsscore}
\end{align*}
Consider the second term. By Corollary~\ref{estimation-error}, the algorithm estimates $|\mc{L}^{\mb{x}_{t-1}}_{\mu,j}|/|C_{\rho}|$ within additive error $\lambda_{\ref{estimation-error}}$ with probability at least $1-\theta_{\ref{estimation-error}}=1-\beta$. Therefore,
\begin{align*}
    \frac{\sum\limits_{j=1}^{k}|\mc{L}^{\mb{x}_{t-1}}_{\mu,j}|(1+\mu)^{\epsscore (j-1)}}{\sum\limits_{j=1}^{k}\tilde{L}^{\mb{x}_{t-1}}_{\mu,j}(1+\mu)^{\epsscore (j-1)}} 
    &\leq \frac{\sum\limits_{j=1}^{k}(\tilde{L}^{\mb{x}_{t-1}}_{\mu,j}+\lambda_{\ref{estimation-error}}|C_{\rho}|)(1+\mu)^{\epsscore (j-1)}}{\sum\limits_{j=1}^{k}\tilde{L}^{\mb{x}_{t-1}}_{\mu,j}(1+\mu)^{\epsscore (j-1)}} \\
    &\leq 1+ \frac{\sum\limits_{j=1}^{k}(\lambda_{\ref{estimation-error}}|C_{\rho}|)(1+\mu)^{\epsscore (j-1)}}{\sum\limits_{j=1}^{k}\tilde{L}^{\mb{x}_{t-1}}_{\mu,j}(1+\mu)^{\epsscore (j-1)}}\\ 
    & \leq 1+\sum\limits_{j=1}^{k}\frac{\lambda_{\ref{estimation-error}}|C_{\rho}|}{\tilde{L}^{\mb{x}_{t-1}}_{\mu,j}}\\
    & \leq 1+k\lambda_{\ref{estimation-error}} |C_{\rho}|
\end{align*}
 Therefore, putting both upper bounds together yields
 \begin{align*}
 \begin{multlined}[t]
     \Pr\left[q(\mb{y}_t) \leq OPT -\frac{2\Delta}{\epsprivacy}\xi\right]
     \leq 
     \exp(-\xi)(1+\mu)^{\epsscore} (1+k\lambda_{\ref{estimation-error}} |C_{\rho}|)
\end{multlined}
 \end{align*}
 As there are at most $|C_{\rho}|$ outputs with quality $OPT-\frac{2\Delta}{\epsprivacy}\xi$ their cumulative probability is at most
\begin{align*}
  |C_{\rho}|(1+k\lambda_{\ref{estimation-error}} |C_{\rho}|){\left(1+\mu\right)}^{\epsscore} \exp(-\xi)
  &=  \frac{|C_{\rho}|(1+k\lambda_{\ref{estimation-error}} |C_{\rho}|){(1+\mu)}^{\epsscore}\beta}{|C_{\rho}|(1+k\lambda_{\ref{estimation-error}} |C_{\rho}|)(1+\mu)^{\epsscore}}\\
  &=\beta.\qedhere
\end{align*}
 


\end{proof}
\subsection{Proof of Theorem~\ref{IDPGA-optimality}}
\begin{proof}[Proof of Theorem~\ref{IDPGA-optimality}]

Suppose $\mb{y}'\in C_{\rho}$ with $\|\mb{y}'-\mb{x}^*\|_2 \leq \rho$. Let $\beta=\frac{1}{|E|^2}$. By Lemma~\ref{lemma:EM-Analogous}, with probability at least $1-\frac{1}{|E|^2}$, we have
  \begin{align*}
    \ip{\mb{y}_{t}, \nabla f(\mb{x}_{t})}
    &\geq \argmax_{\mb{y}\in C_{\rho}}\ip{\mb{y}, \nabla f(\mb{x}_{t})} - \frac{2\Delta}{\epsprivacy}\xi \\
    & \geq \ip{\mb{y}', \nabla f(\mb{x}_{t})} - \frac{2\Delta}{\epsprivacy}\xi\\
    & \geq f(\mb{x}^*) - f(\mb{x}_{t}) - C_{\ref{lem:existance-in-net}}\sqrt{\rho}  - \frac{2 \Delta}{\epsprivacy}\xi\tag{by Lemma~\ref{lem:existance-in-net}}
  \end{align*}
  By a union bound, with probability at least $1-\frac{1}{\mathrm{poly}(|E|)}$, the above inequality holds for every $t$.
  In what follows, we assume this has happened.
  As in the proof of Theorem~\ref{thm:optimality}, suppose $t$ is a continuous variable and define $\frac{d\mb{x}_t}{dt}=\alpha\mb{y}_t$. 
  \begin{align*}
      \frac{d f(\mb{x}_t)}{dt}
      &=\sum_{e}\frac{\partial f(\mb{x}_t(e))}{\partial \mb{x}_t(e)}\frac{d\mb{x}_t(e)}{dt}\\
      &= \nabla f(\mb{x}_t)\cdot \frac{d\mb{x}_t}{dt}\\
      &=\alpha\ip{\mb{y}_t,\nabla f(\mb{x}_t)}\\
      & \geq \alpha\left(f(\mb{x}^*) - f(\mb{x}_{t}) - C_{\ref{lem:existance-in-net}}\sqrt{\rho}  - \frac{2\Delta}{\epsprivacy}\xi\right),
  \end{align*}
  Solving the differential equation with $f(\mb{x}_0)=0$ gives us
   \begin{align*}
      f(\mb{x}_t) = (1-\mathrm{e}^{-\alpha t})\left(f(\mb{x}^*) - C_{\ref{lem:existance-in-net}}\sqrt{\rho}  - \frac{2\Delta}{\epsprivacy}\xi\right).
  \end{align*}
For $\alpha = \frac{1}{T}$ and $t=T$ we obtain
  \begin{align*}
      f(\mb{x}_T) &= (1-\mathrm{e}^{-1})\left(f(\mb{x}^*) - C_{\ref{lem:existance-in-net}}\sqrt{\rho}  - \frac{2\Delta}{\epsprivacy}\xi\right)\\
      &= f(\mb{x}^*) (1-\mathrm{e}^{-1}) - O\left(C_{\ref{lem:existance-in-net}}\sqrt{\rho}  + \frac{2\Delta}{\epsprivacy}\xi\right).
  \end{align*}
Recall that $\xi=\ln\left(\left[|C_{\rho}|(1+k\lambda_{\ref{estimation-error}} |C_{\rho}|)(1+\mu)^{\epsscore}\right]/\beta\right)$ and $\beta= 1/|E|^2$. Next we give an upper bound for the error term.
\begin{align*}
    O\left(C_{\ref{lem:existance-in-net}}\sqrt{\rho}  + \frac{2\Delta}{\epsprivacy}\xi\right)
    &= 
    O\bigg(
        C_{\ref{lem:existance-in-net}}\sqrt{\rho}  + \frac{2\Delta}{\epsprivacy}\ln(|E|^2|C_{\rho}|)+\frac{2\Delta}{\epsprivacy}\ln(1+\mu)^{\epsscore}+
        \frac{2\Delta}{\epsprivacy}\ln(k\lambda_{\ref{estimation-error}} |C_{\rho}|)
    \bigg)
    \\
    &=
    O\bigg(
        C_{\ref{lem:existance-in-net}}\sqrt{\rho}  + \ln(1+\mu)+ \frac{2\Delta}{\epsprivacy}\Big[\ln(|E|^2|C_{\rho}|)+
        \ln(k\lambda_{\ref{estimation-error}} |C_{\rho}|)\Big]
        \bigg)
    \\
    &= 
    O\bigg(
    C_{\ref{lem:existance-in-net}}\sqrt{\rho}  + \ln(1+\mu)+
    \frac{\Delta}{\epsprivacy}(\frac{B}{\rho})^2(\ln|E|+\ln(k\lambda_{\ref{estimation-error}}))
    \bigg)
\end{align*}
Note that by letting $\mu=\mathrm{e}^{\epsprivacy}-1$ we get $\ln(1+\mu)=\epsprivacy$. Moreover, we get $k\leq \frac{r(\mc{M})}{\epsprivacy}$.
\end{proof}

\subsection{Proof of Theorem~\ref{IDPGA-privacy}}
\begin{proof}[Proof of Theorem~\ref{IDPGA-privacy}]
Let $M$ denote Algorithm~\ref{alg:IDPGA}. Let $D$ and $D'$ be two neighboring datasets and $F_D$ and $F_{D'}$ be their associated functions. Suppose $C'_{\rho}(D,t)$ denotes the set of sampled points at time step $t$ given dataset $D$. Similarly, $C'_{\rho}(D',t)$ denotes set of sampled points at time step $t$ given dataset $D'$. Samples are drawn uniformly at random and independent from the input function. Hence, Line~\ref{sampling} of $M$ is $0$-differentially private. Therefore, we assume $C'_{\rho}(D,t)=C'_{\rho}(D',t)=S_t$ for every time step $t$. Define $k, k'$ as follow:
\begin{align*}
  k= \left\lceil \log\limits_{1+\mu}\left(\frac{\max\limits_{\mb{y}\in C_{\rho}}\exp\bigl(\ip{ \mb{y},\nabla f_D(\mb{x})}\bigr)}{\min\limits_{\mb{y}\in C_{\rho}}\exp\bigl(\ip{ \mb{y},\nabla f_D(\mb{x})}\bigr)}\right)\right\rceil 
  \\ 
  k'= \left\lceil \log\limits_{1+\mu}\left(\frac{\max\limits_{\mb{y}\in C_{\rho}}\exp\bigl(\ip{ \mb{y},\nabla f_{D'}(\mb{x})}\bigr)}{\min\limits_{\mb{y}\in C_{\rho}}\exp\bigl(\ip{ \mb{y},\nabla f_{D'}(\mb{x})}\bigr)}\right)\right\rceil 
\end{align*}
 
Note that the layers might be different. Let us use $\mc{L}_i(D)$ and $\mc{L}_i(D')$ for the $i$-th layer given dataset $D$ and $D'$, respectively. Further, $\tilde{L}_i(D)$ and $\tilde{L}_i(D')$ denote the estimated size of the $i$-th layer.  

For a fixed $\mb{y}\in C_{\rho}$, we consider the relative probability of $M$ choosing $\mb{y}$ at time step $t$ given multilinear extensions of $F_D$ and $F_{D'}$. Let $M_t(f_D \mid \mb{x}_t)$ denote the output of $M$ at time step $t$ given dataset $D$ and point $\mb{x}_t$. Similarly, $M_t(f_{D'} \mid \mb{x}_t)$ denote the output of $M$ at time step $t$ given dataset $D'$ and point $\mb{x}_t$. Further, write $d_{\mb{y}}=\ip{\mb{y},\nabla f_D(\mb{x}_t)}$ and  $d'_{\mb{y}}=\ip{\mb{y},\nabla f_{D'}(\mb{x}_t)}$. 

Suppose $\mb{y}\in \mc{L}_i(D)$ given dataset $D$, and $\mb{y}\in \mc{L}_{i'}(D')$ given dataset $D'$. Then, we have

\begin{align}
    \nonumber
    \frac{\Pr[M_t(f_D \mid \mb{x}_t) = \mb{y}]}{\Pr[M_t(f_{D'} \mid \mb{x}_t)=\mb{y}]}
    \nonumber
    &=
    \frac{\Pr[\mb{y}\in S_t \mid D]}{\Pr[\mb{y}\in S_t \mid D']}
    \times \frac{\frac{|\tilde{L}_i(D)|(1+\mu)^{\epsscore(i-1)}}{|\tilde{L}_i(D)|}}{\frac{|\tilde{L}_{i'}(D')|(1+\mu)^{\epsscore(i'-1)}}{|\tilde{L}_{i'}(D')|}}
    \times\frac{\sum\limits_{j=1}^{k'}\tilde{L}_j(D')\exp(\epsscore\cdot (1+\mu)^{j-1})}{\sum\limits_{j=1}^{k}\tilde{L}_j(D)\exp(\epsscore\cdot (1+\mu)^{j-1})}
    \\
    &= \frac{(1+\mu)^{\epsscore(i-1)}}{(1+\mu)^{\epsscore(i'-1)}}
    \times
    \frac{\sum\limits_{j=1}^{k'}\tilde{L}_j(D')\exp(\epsscore\cdot (1+\mu)^{j-1})}{\sum\limits_{j=1}^{k}\tilde{L}_j(D)\exp(\epsscore\cdot (1+\mu)^{j-1})}
    \label{privacy-bound}
\end{align}

The second equality holds since points are sampled uniformly at random from $C_{\rho}$ in Line~\ref{sampling}. 

\begin{lemma}\label{change-of-layer}
Let $D,D'$ be neighboring datasets and $F$ be $\Delta$-sensitive. Suppose $\mb{z}\in C_{\rho}$ is a point in $\mc{L}_j(D)$. Then
\[
\begin{multlined}[t]
     (1+\mu)^{\epsscore(j-1)} \exp(-\frac{\epsprivacy r(\mc{M})}{2})\leq \exp(\epsscore\ip{ \mb{z},\nabla f_{D'}(\mb{x}_t)}) \\
     < (1+\mu)^{\epsscore j} \exp(\frac{\epsprivacy 
     r(\mc{M})}{2})
\end{multlined}
\]
\end{lemma}
\begin{proof}
Since $\mb{z}\in C_{\rho}$ is a point in $\mc{L}_j(D)$, then $(1+\mu)^{j-1}\leq \exp(\ip{ \mb{z},\nabla f_D(\mb{x})}) < (1+\mu)^{j}$. Since $F_D$ is $\Delta$-sensitive hence $f_D$ is $\Delta r(\mc{M})$-sensitive (recall the proof of Theorem~\ref{thm:privacy}). Therefore,
  \begin{align}
      \label{upper-bound}
      &
      \exp(\ip{ \mb{z},\nabla f_{D'}(\mb{x}_t)})\leq \exp(\ip{\mb{z},\nabla f_D(\mb{x}_t)}+\Delta r(\mc{M})) 
      < (1+\mu)^j \exp(\Delta r(\mc{M}))
      \\
      \label{lower-bound}
      &
     (1+\mu)^{j-1} \exp(-\Delta r(\mc{M}))\leq \exp(\ip{ \mb{z},\nabla f_{D}(\mb{x}_t)}-\Delta r(\mc{M})))
     \leq \exp(\ip{\mb{z},\nabla f_{D'}(\mb{x}_t)})
     \\
     \label{up-low}
     &
     (\ref{upper-bound}),(\ref{lower-bound})\Rightarrow 
     (1+\mu)^{j-1} \exp(-\Delta r(\mc{M}))\leq \exp(\ip{ \mb{z},\nabla f_{D'}(\mb{x}_t)}) 
     < (1+\mu)^j \exp(\Delta r(\mc{M}))
     \\
     &
     (\ref{up-low})\Rightarrow 
     (1+\mu)^{\epsscore(j-1)} \exp(-\frac{\epsprivacy r(\mc{M})}{2})\leq \exp(\epsscore\ip{ \mb{z},\nabla f_{D'}(\mb{x}_t)}) 
     < (1+\mu)^{\epsscore j} \exp(\frac{\epsprivacy
     r(\mc{M})}{2})
  \end{align}
\end{proof}

The interpretation of (\ref{up-low}) is that if a point $\mb{z}\in S_t$ appears in layer $\mc{L}_j(D)$ then it can be in any of the layers $\mc{L}_p(D')$ for
\[
(j-1)+\log\limits_{1+\mu}[\exp(-\Delta r(\mc{M}))]\leq p < j+\log\limits_{1+\mu}[\exp({\Delta r(\mc{M})})]\footnote{Note that in low-sensitivity regime, where $\Delta \ll r(\mc{M})$, we have $j-1\leq p <j$.}.
\] 
In a sense, the same argument in Claim~\ref{k-and-k'} shows that $\lfloor\frac{k'}{k}\rfloor=1$. Now, we are ready to provide an upper bound for \eqref{privacy-bound}.

Consider the first term $\frac{(1+\mu)^{\epsscore(i-1)}}{(1+\mu)^{\epsscore(i'-1)}}$. Recall that $\mb{y}\in \mc{L}_i(D)$ given dataset $D$, and $\mb{y}\in \mc{L}_{i'}(D')$ given dataset $D'$. By Lemma~\ref{change-of-layer}, we have 
\[
(1+\mu)^{\epsscore(i-1)} \exp(-\frac{\epsprivacy r(\mc{M})}{2})\leq \exp(\epsscore\ip{ \mb{z},\nabla f_{D'}(\mb{x}_t)}).
\]
Therefore,
\begin{align*}
    \frac{(1+\mu)^{\epsscore(i-1)}}{(1+\mu)^{\epsscore(i'-1)}} &\leq 
    \frac{(1+\mu)^{\epsscore(i-1)}}{(1+\mu)^{\epsscore(i-1)} \exp(-\frac{\epsprivacy r(\mc{M})}{2})}\\
    &= \exp(\frac{\epsprivacy r(\mc{M})}{2})
\end{align*}
Now,  we provide an upper bound for the second term of \eqref{privacy-bound}:
    \begin{align*}
      \frac{\sum\limits_{j=1}^{k'}\tilde{L}_j(D')\exp(\epsscore\cdot (1+\mu)^{j-1})}{\sum\limits_{j=1}^{k}\tilde{L}_j(D)\exp(\epsscore\cdot (1+\mu)^{j-1})}
      &\leq \frac{\sum\limits_{j=1}^{k}\tilde{L}_j(D)\exp(\epsprivacy r(\mc{M}))\exp(\epsscore\cdot (1+\mu)^{j-1})}{\sum\limits_{j=1}^{k}\tilde{L}_j(D)\exp(\epsscore\cdot (1+\mu)^{j-1})}\\
      &=\frac{[\exp(\epsprivacy r(\mc{M}))]\sum\limits_{j=1}^{k}\tilde{L}_j(D)\exp(\epsscore\cdot (1+\mu)^{j-1})}{\sum\limits_{j=1}^{k}\tilde{L}_j(D)\exp(\epsscore\cdot (1+\mu)^{j-1})}\\
      &\leq \exp(\epsprivacy r(\mc{M}))
    \end{align*}
By a union bound and composition Theorem~\ref{k-fold-composition}, Algorithm~\ref{alg:IDPGA} preserves $O({\epsprivacy T r(\mathcal{M})})\leq O(\epsprivacy r(\mathcal{M})^2)$-differential privacy. The heart of the above inequality is that, given the set of sample points, the layers defined for both instances are almost identical. 
\end{proof}

\begin{claim}\label{k-and-k'}
  $\frac{k'}{k}\leq 1+\frac{2\Delta r(\mc{M})}{k\ln(1+\mu)}$.
\end{claim}
\begin{proof}
  \begin{align*}
    k' &=  \log\limits_{1+\mu}\left(\frac{\max\limits_{\mb{y}\in C_{\rho}}\exp\bigl(\ip{ \mb{y},\nabla f_{D'}(\mb{x})}\bigr)}{\min\limits_{\mb{y}\in C_{\rho}}\exp\bigl(\ip{ \mb{y},\nabla f_{D'}(\mb{x})}\bigr)}\right)  \\
    & \leq  
    \log\limits_{1+\mu}\left(\frac{\max\limits_{\mb{y}\in C_{\rho}}\exp\bigl(\ip{ \mb{y},\nabla f_{D}(\mb{x})}+\Delta r(\mc{M})\bigr)}{\min\limits_{\mb{y}\in C_{\rho}}\exp\bigl(\ip{ \mb{y},\nabla f_{D}(\mb{x})}-\Delta r(\mc{M})\bigr)}\right) \\
    & \leq  
    \log\limits_{1+\mu}\left(\frac{\max\limits_{\mb{y}\in C_{\rho}}\exp\bigl(\ip{ \mb{y},\nabla f_{D}(\mb{x})}\bigr)}{\min\limits_{\mb{y}\in C_{\rho}}\exp\bigl(\ip{ \mb{y},\nabla f_{D}(\mb{x})}\bigr)}\right)\\
    &+\log\limits_{1+\mu}\exp(2\Delta r(\mc{M}))\\
    &= k + \frac{2\Delta r(\mc{M})}{\ln(1+\mu)} \\
    & = k\left(1+\frac{2\Delta r(\mc{M})}{k\ln(1+\mu)}\right). \qedhere
  \end{align*}
\end{proof}

\section{Missing Proofs from Section \ref{section;k-submodular}}
\subsection{Proof of Theorem~\ref{opt-k-submodular}}
\begin{proof}[Proof of Theorem~\ref{opt-k-submodular}]
 Consider the $j$-th iteration of the algorithm. Let $(e^{(j)},i^{(j)})$ be the pair chosen in this iteration. Further, let $\mb{o}$ be the optimal solution and $\mb{x}^{(j)}$ be the solution after the $j$-th iteration. Note that $|\mathrm{supp}(\mb{x}^{(j)})|=j$ for $j\in[r(\mc{M})]$. We define a sequence of vectors $\mb{o}^{(0)}=\mb{o},\mb{o}^{(1)},\dots,\mb{o}^{r(\mc{M})}$, as in~\cite{iwata2016improved,ohsaka2015monotone,sakaue2017maximizing,ward2014maximizing}, such that
 \begin{itemize}
     \item[1.] $\mb{x}^{(j)}\prec \mb{o}^{(j)}$ for all $0\leq j\leq r(\mc{M})-1$,
     \item[2.] $\mb{x}^{(r(\mc{M}))}= \mb{o}^{(r(\mc{M}))}$,
     \item[3.] $O^{(j)}:=\mathrm{supp}(\mb{o}^{(j)})\in\mc{B}$ for all $0\leq j\leq r(\mc{M})$.
 \end{itemize}
 For the sake of completeness, let us describe how to obtain $\mb{o}^{(j)}$ from $\mb{o}^{(j-1)}$ assuming $\mb{x}^{(j-1)}\prec \mb{o}^{(j-1)}$ and $O^{(j-1)}\in\mc{B}$.
 Let $X^{(j)}=\mathrm{supp}(\mb{x}^{(j)})$.
 $\mb{x}^{(j-1)}\prec \mb{o}^{(j-1)}$ implies that $X^{(j-1)}\subsetneq O^{(j-1)}$ and $e^{(j)}$ is chosen to satisfy $X^{(j-1)}\cup\{e^{(j)}\}\in\mc{I}$.
 By Lemma~\ref{extension}, there exists $e'\in O^{(j-1)}\setminus X^{(j-1)}$ such that $O^{(j-1)}\setminus \{e'\} \cup\{e^{(j)}\}\in\mc{B}$.

Now let $o^{(j)} = e'$ and define $\mb{o}^{(j-1/2)}$ as the vector obtained by assigning 0 to the $o^{(j)}$-th element of $\mb{o}^{(j-1)}$.
We then define $\mb{o}^{(j)}$ as the vector obtained from $\mb{o}^{(j-1/2)}$ by assigning $i^{(j)}$ to the $e^{(j)}$-th element.
Therefore, for vector $\mb{o}^{(j)}$ we have $O^{(j)}\in \mc{B}$ and $\mb{x}^{(j)}\prec \mb{o}^{(j)}$.

 By Theorem 2 in~\cite{sakaue2017maximizing}, if we always selected $(e^{(j)},i^{(j)})$ with $e^{(j)}\in \Lambda(\mb{x}),i\in[k]$ and maximum $\Delta_{e,i}f(\mb{x})$, we would have
 \[
 F(\mb{x}^{(j)})-F(\mb{x}^{(j-1)})\geq F(\mb{o}^{(j-1)})-F(\mb{o}^{(j)}).
 \]
 Instead we use the exponential mechanism which, by Theorem~\ref{thm:EM-bound}, selects $(e^{(j)},i^{(j)})$ within $ \frac{2\Delta}{\epsprivacy}\ln{\frac{|\Lambda(\mb{x}^{(j)})|}{\beta}}$ from the optimal choice with probability at least $1-\beta$. Therefore,
 \[
 \begin{multlined}[t]
  F(\mb{x}^{(j)})-F(\mb{x}^{(j-1)})\geq  F(\mb{o}^{(j-1)})-F(\mb{o}^{(j)})-\frac{2\Delta}{\epsprivacy}\ln{\frac{|\Lambda(\mb{x}^{(j)})|}{\beta}}
 \end{multlined}
 \]
 with probability at least $1-\beta$.
 Given this, one can derive the following:
\begin{align*}
 F(\mb{o})-F(\mb{x}^{(r(\mc{M}))})
 &=\sum\limits_{j=1}^{r(\mc{M})} F(\mb{o}^{(j-1)})-F(\mb{o}^{(j)})\\
 &\leq \sum\limits_{j=1}^{r(\mc{M})}\left( F(\mb{x}^{(j-1)})-F(\mb{x}^{(j)})+\frac{2\Delta}{\epsprivacy}\ln{\frac{|\Lambda(\mb{x}^{(j)})|}{\beta}}\right)\\
 &=F(\mb{x}^{(r(\mc{M}))})-F(\mb{0})+r(\mc{M})\left(\frac{2\Delta}{\epsprivacy}\ln{\frac{|\Lambda(\mb{x}^{(j)})|}{\beta}}\right)\\
 &=F(\mb{x}^{(r(\mc{M}))})+r(\mc{M})\left(\frac{2\Delta}{\epsprivacy}\ln{\frac{|\Lambda(\mb{x}^{(j)})|}{\beta}}\right),
\end{align*}
which means Algorithm~\ref{alg:DP-k-submodular-matroid} returns $\mb{x}=\mb{x}^{(r(\mc{M}))}$ with quality at least $\frac{1}{2}\mathrm{OPT}-r(\mc{M})(\frac{2\Delta}{\epsprivacy}$ $\ln{\frac{|\Lambda(\mb{x}^{(j)})|}{\beta}})$ with probability at least $1-r(\mc{M})\beta$. Having $\beta=\frac{1}{|E|^2}$, $|\Lambda(\mb{x}^{(j)})|\leq |E|$ gives us 
\[F(\mb{x}) \geq \frac{1}{2}\mathrm{OPT}-O\left(\frac{\Delta r(\mc{M})\ln{|E|}}{\epsprivacy}\right).
\qedhere \]
\end{proof}

\subsection{Proof of Theorem~\ref{Imp-k-sub-opt}}
\begin{proof}[Proof of Theorem~\ref{Imp-k-sub-opt}]

Let $R^{(j)}$ be $R$ in the $j$-th iteration, $\mb{o}$ be the optimal solution and $\mb{x}^{(j)}$ be the solution after the $j$-th iteration.
Further, let $X^{(j)}=\mathrm{supp}(\mb{x}^{(j)})$, $O^{(j)}=\mathrm{supp}(\mb{o}^{(j)})$, and
\[{\Lambda(\mb{x})}^{(j)}=\{e\in E\setminus \mathrm{supp}(\mb{x}^{(j)})\mid \mathrm{supp}(\mb{x}^{(j)})\cup\{e\}\in \mc{I}\}
\]

We iteratively define $\mb{o}^{(0)}=\mb{o},\mb{o}^{(1)},\dots,\mb{o}^{r(\mc{M})}$ as follows. If  $R^{(j)}\cap {\Lambda(\mb{x})}^{(j)}=\emptyset$, then we regard that the algorithm failed. Else we proceed as follows.
By Lemma~\ref{extension}, for any $e^{(j)}\in R^{(j)}\cap {\Lambda(\mb{x})}^{(j)}$, there exists $e'$ such that $e'\in O^{(j-1)}\setminus X^{(j-1)}$ and $O^{(j-1)}\setminus \{e'\} \cup\{e^{(j)}\}\in\mc{B}$. Now let $o^{(j)} = e'$ and define $\mb{o}^{(j-1/2)}$ as the vector obtained by assigning 0 to the $o^{(j)}$-th element of $\mb{o}^{(j-1)}$. We then define $\mb{o}^{(j)}$ as the vector obtained from $\mb{o}^{(j-1/2)}$ by assigning $i^{(j)}$ to the $e^{(j)}$-th element. Therefore, for vector $\mb{o}^{(j)}$ we have $O^{(j)}\in \mc{B}$ and $\mb{x}^{(j)}\prec \mb{o}^{(j)}$.

If the algorithm does not fail and $\mb{o}^{(0)}=\mb{o},\mb{o}^{(1)},\dots,\mb{o}^{r(\mc{M})}$ are well defined, or in other words, if $R^{(j)}\cap {\Lambda(\mb{x})}^{(j)}$ is not empty for every $j \in [r(\mc{M})]$, then the rest of the analysis is completely the same as in Theorem~\ref{opt-k-submodular}, and we achieve an approximation ratio of (roughly) $1/2$. Hence, it suffices to show that $R^{(j)}\cap {\Lambda(\mb{x})}^{(j)}$ is not empty with a high probability.



\begin{lemma}\label{k-sub-Imp}
With probability at least $1-\frac{\gamma}{r(\mc{M})}$, we have $R^{(j)}\cap {\Lambda(\mb{x})}^{(j)}\neq\emptyset$ for every $j\in[r(\mc{M})]$.
\end{lemma}
Analogous to the analysis in Theorem~\ref{opt-k-submodular}, for every time step $0\leq j\leq r(\mc{M})$, with probability at least $1-\frac{\gamma}{r(\mc{M})}$ we have 
\[
  F(\mb{x}^{(j)})-F(\mb{x}^{(j-1)})\geq F(\mb{o}^{(j-1)})-F(\mb{o}^{(j)})-\frac{2\Delta}{\epsprivacy}\ln({\frac{r(\mc{M})|\Lambda(\mb{x}^{(j)})|}{\gamma}})
 .
 \]
 By a union bound over $j \in [r(\mc{M})]$, with probability at least $1-\gamma$, it follows that
 \[
 F(\mb{x}) \geq \frac{1}{2}\mathrm{OPT}-O\left(\frac{\Delta r(\mc{M})\ln(|E|/\gamma)}{\epsprivacy}\right).
 \]
  Applying a similar argument as in~\cite{ohsaka2015monotone}, the number of evaluations of $f$ is at most
 \begin{align*}
     k\sum\limits_{t=1}^{r(\mc{M})} \frac{|E|-t+1}{r(\mc{M})-t+1} \ln\frac{r(\mc{M})}{\gamma}
     &=k\sum\limits_{t=1}^{r(\mc{M})} \frac{|E|-r(\mc{M})+t}{t} \log\frac{r(\mc{M})}{\gamma}\\
     &=O\left(k|E|\ln{r(\mc{M})} \ln{\frac{r(\mc{M})}{\gamma}} \right)
     \qedhere
 \end{align*}
\end{proof}

\begin{proof}[Proof of Lemma~\ref{k-sub-Imp}]
\begin{align*}
    \Pr[R^{(j)}\cap {\Lambda(\mb{x})}^{(j)}=\emptyset]
    &={\left(1-\frac{r(\mc{M})-\mathrm{supp}(\mb{x}^{(j)})}{|E\setminus \mathrm{supp}(\mb{x}^{(j)})|}\right)}^{|R^{(j)}|}\\
    &\leq \exp\left(-\frac{r(\mc{M})-j+1}{|E|-j+1}\frac{|E|-j+1}{r(\mc{M})-j+1}\ln\frac{r(\mc{M})}{\gamma}\right)\\
    &=\exp\left(-\ln\frac{r(\mc{M})}{\gamma}\right)=\frac{\gamma}{r(\mc{M})}
    \qedhere
\end{align*}
\end{proof}